\documentclass[a4paper,12pt]{article}
\pdfoutput=1 

\usepackage[utf8]{inputenc}
\usepackage[T1]{fontenc}
\usepackage{titling}

\usepackage{amsmath, amssymb}
\usepackage{pgf}
\usepackage{amsfonts}
\usepackage{dsfont}
\usepackage{mathrsfs}
\usepackage{mathabx}
\usepackage{stmaryrd}
\usepackage{makeidx}
\usepackage{amsbsy}
\usepackage{amsthm}
\usepackage{color}
\usepackage{fullpage}
\usepackage{enumitem}
\usepackage[vcentermath]{youngtab}
\usepackage{marginnote} 
\usepackage{mdframed}	
\newmdenv[topline=false, bottomline=false, skipabove=\topsep, skipbelow=\topsep]{siderules}
\usepackage{physics}
\usepackage{bbm}
\usepackage{enumitem}
\usepackage{mathrsfs}

\newtheorem{theorem}{Theorem}
\newtheorem{lemma}{Lemma}
\newtheorem*{lemma*}{Lemma}
\newtheorem{corollary}{Corollary}
\newtheorem{conjecture}{Conjecture}

\theoremstyle{definition}
\newtheorem{definition}{Definition}
\theoremstyle{remark}
\newtheorem{remark}{Remark}

\newcommand{\ca}[1]{{\mathcal #1}}
\newcommand{\ben}{\begin{equation}}
\newcommand{\een}{\end{equation}}
\def\bena{\begin{eqnarray}}
\def\eena{\end{eqnarray}}

\def\cA{{\mathcal A}}
\def\cB{{\mathcal B}}

\def\cV{{\ca V}}

\def\sC{\mathscr{C}}

\def\sH{\mathscr{H}}

\def\sP{\mathscr{P}}

\def\sR{\mathscr{R}}
\def\sS{\mathscr{S}}

\def\sY{\mathscr{Y}}

\def\bS{{\mathbb S}}

\def\cA{{\mathcal A}}
\def\cB{{\mathcal B}}

\newcommand{\rmo}[2]{\Delta_{#1, #2}}
\newcommand{\mo}[1]{\Delta_{#1}}
\newcommand{\mA}{\mathcal{A}}
\newcommand{\mB}{\mathcal{B}}

\def\({\left(}
\def\){\right)}
\def\ud{{\rm d}}
\newcommand{\beq}{\begin{equation}}
\newcommand{\eeq}{\end{equation}}

\begin{document}
\title{Approximate recovery and relative entropy I. 
general von Neumann subalgebras}

	\author{Thomas Faulkner$^{1}$, Stefan Hollands$^{2}$, Brian Swingle$^{3}$, Yixu Wang$^{3}$\\
	{\it $^1$ University of Illinois at Urbana-Champaign, IL and KITP, Santa Barbara}\\
	{\it $^2$ ITP, Universit\" at Leipzig, MPI-MiS Leipzig, and KITP, Santa Barbara}\\
	{\it $^3$ Maryland Center for Fundamental Physics} \\ 
	{\it and University of Maryland, College Park, MD, USA}
	}

\date{\today}
	
\maketitle

\begin{abstract}
We prove the existence of a universal recovery channel that approximately recovers states on a v. Neumann subalgebra when the change in relative entropy,
with respect to a fixed reference state, is small. Our result is a generalization of previous results that 
applied to type-I v. Neumann algebras by Junge at al. [arXiv:1509.07127]. 
We broadly follow their proof strategy but consider here arbitrary v. Neumann algebras, where qualitatively new issues arise. 
Our results hinge on the construction of certain analytic vectors and computations/estimations of their Araki-Masuda $L_p$ norms.    
We comment on applications to the quantum null energy condition. 
\end{abstract}

\section{Introduction}

Quantum error correction is an important tool in quantum computation but has physical manifestations well beyond this domain. For example, it has become influential in the study of topological aspects of many-body quantum physics~\cite{Wen-1991,Kitaev-2006,Dennis-2002}, renormalization group approaches to interacting theories~\cite{Swingle-2016,Kim-2017}, random quantum systems~\cite{Brown-2013}, and even basic aspects of quantum gravity in the AdS/CFT correspondence \cite{Almheiri-2014,Faulkner,swingle}. 
While quantum computers typically manipulate finite dimensional Hilbert spaces, many applications of error correction to field theory and gravity go beyond this simple setting and a general treatment requires more sophisticated tools, including tools from the theory of operator algebras. While one might hope to approximate any of these physical systems by finite quantum systems, this point of view can obscure crucial physical features that are more naturally expressed in a less restrictive approach. We will give an example of this in the context of quantum field theory, where operator algebraic approaches have a long tradition, see e.g. \cite{haag_2}. 

At the same time, the operator algebra approach is so general that expressing proofs of fundamental quantum information results in this language exposes the core nature of such proofs and ends up simplifying the approach in many situations. Indeed, many of the original theorems in quantum information have their origin in the study of operator algebras. In this paper, we generalize the results of \cite{Junge}, pertaining to the approximate reversibility of quantum channels, from a type-I v. Neumann algebra\footnote{Direct sums of matrix algebras or the algebra of all bounded operators on a separable Hilbert space.} setting to general v. Neumann algebras (Theorem 2). At the heart of these results is a strengthened version of the monotonicity \cite{Uhlmann1977} of relative entropy (Theorem 1).
In the present paper (part I), we treat the sub-algebra case which involves a simple quantum channel called an inclusion. In a follow-up paper (part II), we treat the general quantum channel case.

Along the way, we prove two theorems that might be of independent interest. The first (Theorem 3) concerns the computations of the derivatives of the ``sandwiched'' and ``Petz'' relative Renyi entropies for two nearby states. We call this result a first law because of its similarity to the first law of black hole thermodynamics in the setting of AdS/CFT \cite{Blanco:2013joa,Faulkner:2013ica}. The second (Theorem 4) pertains to a regularization procedure for relative entropy that produces states with finite relative entropy and also allows for continuous extrapolation of relative entropy when removing the regulator. The vectors that result from this procedure are important here because they lead to extended domains of holomorphy that allow us to proceed towards the proof of strengthened monotonicity with a similar argument as in the finite dimensional setting. 

We will also discuss an application to the study of the quantum information aspects of quantum field theory that requires this general v. Neumann algebra setting. In the field theory context, new results using operator algebra methods have made it possible to make rigorous statements about the dynamics of interacting theories. For example, we propose that the quantum null energy condition, a bound on the local energy density (that has already been proven with other methods~\cite{Bousso-2015,Balakrishnan-2017}), is tightly linked to the strengthened monotonicity result that we derive in this paper. 

\medskip
\noindent
{\bf Notations and conventions:} Calligraphic letters $\mathcal{A,M}, \dots$ denote v. Neumann algebras. Calligraphic letters $\mathscr{H,K}, \dots$ denote more general linear spaces or subsets thereof. $\bS_a=\{z \in {\mathbb C} \mid 0 < \Re (z) < a\}$ denotes an open strip, and we often write $\bS=\bS_1$. We typically use the physicist's ``ket''-notation $|\psi\rangle$ for vectors in a Hilbert space. 
The scalar product is written
\begin{equation}
(|\psi\rangle, |\psi'\rangle)_{\sH} = \langle \psi | \psi' \rangle
\end{equation}
and is anti-linear in the first entry. The norm of a vector is sometimes written simply as
$\| |\psi\rangle \| =: \| \psi \|$. The action of a linear operator $T$ on a ket is sometimes written as $T|\phi\rangle = |T\phi\rangle$. 
In this spirit, the norm of a bounded linear operator $T$ on $\sH$ is written as $\|T\|= \sup_{|\psi\rangle: \|\psi\|=1} \|T\psi\|$.

\section{Basic definitions and main results}
\label{sec:above}

\subsection{Tomita-Takesaki theory}

Here we outline some elements of v. Neumann algebra theory relevant for this work; for details, see \cite{Bratteli,Takesaki,Zsido}.
A v. Neumann algebra, $\mathcal{A}$, is a subspace of the set of all bounded operators $B(\sH)$ containing the unit
operator $1$ that 
is closed under: products, the star operation denoted $a^*$ and limits in the ultra-weak operator topology.
States on $\mathcal{A}$ are linear functionals that are positive, $\rho(a^*a) \ge 0$, normalized, $\rho(1)=1$,
and ``normal'' i.e., continuous in the ultra-weak operator topology. 
The set of normal states is contained in the ``predual'' $\mathcal{A}_\star$ of $\cA$, i.e. the set of all ultra-weakly 
continuous linear functionals on $\cA$. 
One defines the support projection $\pi^\cA$ associated to a state $\rho$ as the smallest projection $\pi = \pi^\cA(\rho)$ in $\mathcal{A}$ that satisfies $\rho(\pi) = 1$. Faithful states by definition have unit support projection. 

We will work with the v. Neumann algebra in a so called standard form, $(\mathcal{A}, \mathscr{H}, J, \mathscr{P}^\natural)$, where $\mathcal{A}$ acts on the Hilbert space $\mathscr{H}$ and where there is an anti-linear, unitary involution $J$ and a self-dual ``natural'' cone $\mathscr{P}^\natural$ left invariant by $J$. The existence and detailed properties of a normal form 
are proven in \cite{haagerup}; here we only mention: One has $ J \mathcal{A} J = \mathcal{A}'$ where $\mathcal{A}' \subset B(\mathscr{H})$, the ``commutant'', is the v. Neumann algebra of all bounded operators on $\mathscr{H}$ that commute with $\mathcal{A}$. The natural cone defines a set of vectors in the Hilbert space that canonically represent states on $\mathcal{A}$ via
\begin{equation}
\cA_\star \owns \rho \mapsto \left| \xi_\rho \right> \in \mathscr{P}^\natural \,, \qquad \rho( \cdot) = \omega_{\xi_\rho}( \cdot ) \equiv \left< \xi_\rho \right| \cdot \left| \xi_{\rho} \right>
\end{equation}
and where we use the notation $ \omega_\psi ( \cdot) \equiv \left< \psi \right| \cdot \left| \psi \right> \in \mathcal{A}_\star$ for the linear functional on $\mathcal{A}$ induced by a vector $\psi \in \mathscr{H}$. The vector in the natural cone representing $\omega_\psi$ will also be denoted by $|\xi_\psi\rangle$. It is known that it is related to $|\psi\rangle$
by a partial isometry $v'_\psi\in \cA'$,
\ben
|\xi_\psi\rangle = v'_\psi |\psi\rangle.
\een
Furthermore,  it is known that\footnote{For the case of matrix algebras, 
the second inequality is known as the Powers-St\" ormer inequality.} proximity of the state functionals implies that of the 
vector representatives in the natural cone and vice versa, in the sense that
\ben\label{eq:PS}
\|\xi_\phi-\xi_\psi\| \,  \|\xi_\phi+\xi_\psi\|  \ge \|\omega_\phi-\omega_\psi\| \ge \|\xi_\phi-\xi_\psi\|^2 
\een
holds.

We now introduce the modular operators that are central to our discussion of relative entropy \cite{Araki1,Araki2} and non-commutative $L_p$-spaces \cite{AM}. This is most straightforward if we have cyclic and separating vector $|\eta\rangle$ for $\mathcal{A}$ algebra, meaning that $\{a |\eta\rangle: a \in \cA\}$ is dense in $\sH$ and that $a|\eta\rangle=0$ implies that $a=0$. 
Then Tomita-Takesaki theory establishes that one can define an anti-linear, unitary operator $J$ and a positive, self-adjoint operator $\Delta_\eta$ by
the relations 
\begin{equation}
J \Delta_\eta^{1/2} a \left| \eta \right> = a^* \left| \eta \right> \,,\quad \forall a \in \mathcal{A} \,
\end{equation}
$\Delta_\eta$ is in general unbounded. $J$ can be used in this case to define a standard form, with 
$\sP^\natural$ given by the closure of $\{aJaJ |\eta \rangle: a \in \cA\}$, but we emphasize that a standard form exists generally even without a faithful state $|\eta\rangle$. From now on, we regard such a standard form, hence $J$, as fixed.
We will continue to take $\eta \in \sP^\natural$.

We will also need the concept of relative modular operator $\Delta_{\phi,\psi}$ \cite{Araki1}. 
In a slight generalization of the above definitions, let $|\phi \rangle, |\psi \rangle \in \sP^\natural$. Then
there is a non-negative, self-adjoint operator $\Delta_{\phi,\psi}$ characterized by
\ben 
\label{modA}
J \Delta_{\phi,\psi}^{1/2}\left( a \left| \psi \right> + \left| \chi \right> \right) = \pi^\cA(\psi) a^* \left| \phi \right>\,,  \quad \forall \,\, a \in \mathcal{A} \,,\,\, \left| \chi \right> \in (1-\pi^{\cA'}(\psi)) \mathscr{H} 
\een
The non-zero support of $\Delta_{\phi,\psi}$ is $\pi^\cA(\phi) \pi^{\cA}(\psi) \mathscr{H}$, and the functions $\Delta_{\phi,\psi}^z$
are understood via the functional calculus on this support and are defined as $0$ on $1-\pi^\cA(\phi) \pi^{\cA}(\psi)$. 
We can similarly define relative modular operators for vectors outside of the natural cone, for a detailed discussion of such matters see 
e.g., \cite{AM}, app. C. 
For example, we may use the well known transformation property of the modular operators $\Delta_{u'\phi,v'\psi}= v' \Delta_{\phi,\psi} {v'}^*$ where
$v',u'\in \cA'$ is a partial isometry (with appropriate initial and final support), to define:
\begin{equation}
\label{eq:change}
\Delta_{\phi,\psi} \equiv {v'_\psi}^* \Delta_{\xi_\phi,\xi_\psi} v'_\psi \,, \qquad |\psi \rangle, |\phi \rangle \in \mathscr{H}.
\end{equation}
Similarly we can define the relative modular operators for the commutant in direct analogy. We will often denote it by $\Delta_{\phi,\psi}'$. 

When $\ket{\psi}=\ket{\phi}$ we will denote these operators as $\Delta_{\phi,\phi}\equiv \Delta_\phi$. 
This is the non-relative modular operator already discussed from which we can define modular flow:
\begin{equation}
\varsigma_{\phi}^t(a) = \Delta_{\phi}^{it} a \Delta_{\phi}^{-it} \in \mathcal{A}\,, 
\end{equation}
where $a \in \mathcal{A}$ and we have taken $\phi$ to be cyclic and separating.  The modular flow can also be extracted from the relative modular operators:
\begin{equation}
\label{modflowrel}
 \Delta_{\phi,\psi}^{it} a \Delta_{\phi,\psi}^{-it}  = \varsigma_{\phi}^t(a) \pi^{\cA'}(\psi)
\end{equation}
for any $\psi \in \mathscr{H}$. 

The modular operators satisfy various relations that we need to draw on below and we simply quote these here (recall that $\eta \in \mathscr{P}^\natural$):
\begin{equation}
\label{rels}
\Delta_{\psi,\eta}^{-z} =  (\Delta_{\eta,\psi}')^{z}\, , \qquad J \Delta_{\xi_\psi,\eta}^{-z} = \Delta_{\eta,\xi_\psi}^{\bar z}  J 
\,, \qquad \Delta_{\psi,\eta}^{-it} a   \Delta_{\eta}^{it} \in \mathcal{A}
\end{equation}
for $t\in \mathbb{R}$, $z \in \mathbb{C}$ and $a\in \mathcal{A}$ and where these equations make sense when acting on vectors in appropriate domains -- we are more specific about this when we get to use these equations. The Connes cocycle $(D\psi:D\phi)_{t}$ is the partial isometry from $\cA$ defined by $(t \in \mathbb{R})$
\begin{equation}
\label{coswitch}
(\Delta_{\psi,\phi}^{-it}    \Delta_{\phi}^{it}) = \Delta_{\psi,\eta}^{-it}    \Delta_{\phi,\eta}^{it}
\pi^{\cA'}(\phi)
\equiv (D\psi:D\phi)_{-t}
\pi^{\cA'}(\phi).
\end{equation}
According to \cite{Araki1,Araki2}, if $\pi^\cA(\phi) \ge \pi^\cA(\psi)$, the relative entropy may be defined as 
\ben
S(\psi | \phi) = -\lim_{\alpha \to 0^+} \frac{\langle \psi | \Delta^\alpha_{\phi, \psi} \psi \rangle-1}{\alpha} ,  
\een
otherwise, it is by definition equal to $+\infty$. The relative entropy only depends on the 
functionals $\omega_\psi, \omega_\phi$ but not on the particular choice of vectors that define them.

\subsection{Inclusions of v. Neumann algebras and Petz map}

Now consider a v. Neumann subalgebra $\mathcal{B}$ of $\mathcal{A}$.
It is convenient to take $\mathcal{B}$ to be in a 
standard form $(\mathcal{B}, \mathscr{K}, J_{\mathcal{B}}, \mathscr{P}^\natural_{\mathcal{B}})$.
In this representation $\mathcal{B}$ acts on a (potentially) different Hilbert space $\mathscr{K}$ and to distinguish these representations we define the embedding $\iota: \mathcal{B} \rightarrow \mathcal{A}$ as a $*$-isomorphism of v. Neumann algebras from $\mathcal{B}$ to the range $\iota(\mathcal{B}) \subset \mathcal{A}$.

Normal states $\rho$ on $\mathcal{B}$ are induced from states on $\mathcal{A}$ in the obvious way: $\rho|_{\mathcal{B}} \equiv \rho \circ \iota \equiv \iota^+ \rho$, so $\iota^+(\mathcal{A}_\star) \subset \mathcal{B}_\star$. We adopt the convention that the corresponding support projection will be labelled in the following manner:
\begin{equation}
\pi^{\mathcal{B}}(\rho) \equiv  \pi^{\mathcal{B}}(\rho \circ \iota) \,, \quad \rho \in \mathcal{A}_\star
\end{equation}
and we have 
\begin{equation}
\pi^\cA(\rho) \leq  \iota(\pi^{\mathcal{B}}(\rho)), 
\end{equation}
where for two self-adjoint elements $a,b \in \cA$ we say that $a \le b$ if $a-b$ is a non-negative operator. Given 
$\rho, \sigma \in \cA_\star$, we may define the relative entropy $S_\cA(\rho | \sigma) \equiv S(\rho | \sigma)$ as above, and we put
\ben
\label{eq:SB}
S_\cB(\rho | \sigma) \equiv S(\rho \circ \iota | \sigma \circ \iota).
\een
By monotonicity of the relative entropy \cite{Uhlmann1977}, we have $S_{\cA}(\rho| \sigma ) -  S_{\cB}(\rho| \sigma ) \ge 0$.

Given a faithful state $\sigma \in \mathcal{A}_\star$, an isometry $V_\sigma: \mathscr{K} \rightarrow \mathscr{H}$ can be naturally defined as follows \cite{Petz3,Petz4,Petz1993}:
\begin{equation}
\label{embed}
V_\sigma  b \big| \xi_{\sigma}^{\mathcal{B}} \big>  := \iota( b ) \left| \xi_\sigma^\cA \right> \,, \qquad b\in \mathcal{B} \ ,
\end{equation}
where we use the notation $|\xi_{\sigma}^{\mathcal{B}} \rangle$ for the vector representative of the state $\sigma \circ \iota \in \mathcal{B}_\star $ in the natural cone of the algebra $\mathcal{B}$ and $|\xi_{\sigma}^{\mathcal{A}} \rangle$ for the vector representative of the state $\sigma \in \mathcal{A}_\star $ in the natural cone of the algebra $\mathcal{A}$.
As reviewed in Appendix~\ref{app:Vsigma}, this embedding $V_\sigma$ commutes with the action of $b$,
\begin{equation}
\label{embedcomm}
V_{\sigma} \left(b \left| \chi \right> \right) = \iota(b) V_\sigma \left| \chi \right> \,, \qquad \chi \in \mathscr{K} \,, \qquad b \in \mathcal{B}
\end{equation}
and satisfies $V_\sigma^* \iota(b) V_\sigma^{} = b$ for all $b\in \mathcal{B}$ as well as $V_\sigma(\mathscr{K}) = \pi_{\mathscr{K}} \mathscr{H}$ for some projector $V_{\sigma} V_{\sigma}^* \equiv \pi_{\mathscr{K}} \in \iota(\mathcal{B})'$.

We now recall the concept of approximate sufficiency. First, recall that a linear mapping $\alpha: \cA \to \cB$ is called  ``channel'' if it is completely positive, ultra-weakly continuous and $\alpha(1)=1$, see \cite{Petz1993}. 

\begin{definition}
\label{def:suff}
Following \cite{Petz2,Petz1993} we say that the inclusion $\mathcal{B} \subset \mathcal{A}$ is $\epsilon$-\emph{approximately sufficient} for a set of states $\sS \subset \mathcal{A}_\star$, if there exists a fixed  channel
\begin{equation}
\label{recab}
\alpha : \mathcal{A} \rightarrow \mathcal{B},
\end{equation}
called ``\emph{recovery channel} '', 
for which the recovered state is close to the original state in the sense that
\begin{equation}
\| \rho -\rho \circ \iota \circ \alpha \| \equiv \sup_{ a \in \mathcal{A}: \| a \| \leq 1} \left| \rho(a) - \rho \circ \iota \circ \alpha (a) \right| \leq \epsilon , , \quad \forall \,\, \rho \in \sS. 
\end{equation}
Here we take all $\rho \in \sS$ to be normalized $ \rho(1)=1$. 
\end{definition}

Note that if $\mathcal{A} \subset \mathcal{B}$ is $\epsilon$-sufficient for $\sS$, then $\mathcal{A} \subset \mathcal{B}$ is $\epsilon$-sufficient for the closed convex hull
of states $\overline{ {\rm conv} (\sS)}$.

We would now like to construct an $\alpha$ that works as a recovery map for a set of states that are close in relative entropy under restriction to the sub-algebra.  We take the relative entropy to compare to a fixed state $\sigma \in \mathcal{A}_\star$. That is, we consider the set
\begin{equation}
\label{Rsigma}
\sR^{(\sigma)}_{\delta} = \left\{ \rho \in \mathcal{A}_\star :  \rho(1) = 1\,,\,\, \rho \geq 0 \,, \,\, S_{\cA}(\rho| \sigma ) -  S_{\cB}(\rho| \sigma ) \le \delta \right\}
\end{equation}

The required recovery channel is related to the so-called Petz map, which is defined in the sub-algebra context (and faithful $\sigma$) as (see e.g., \cite{Petz1993}, sec. 8):
\begin{equation}
\alpha_\sigma(\cdot) = J_{\mathcal{B} } V_\sigma^* J_\cA \left( \cdot \right) J_\cA V_\sigma J_{\mathcal{B} }
\end{equation}
It maps operators on $\mathscr{H}$ to operators on $\mathscr{K}$, and furthermore
\begin{equation}
\alpha_\sigma(\mathcal{A}) \subset \mathcal{B}.
\end{equation}
As shown in \cite{Petz1993}, prop. 8.3 this map satisfies the defining properties of a recovery channel used in def.~\ref{def:suff} -- in fact, in the subalgebra context considered here 
it is equal to the generalized conditional expectation introduced even earlier by \cite{AC}. In the non-faithful case there is a slightly more complicated expression that we will discuss below in lem. \ref{lem:proj}.

\subsection{Main theorems}

Given two states $\rho, \sigma \in \cA_\star$, the fidelity is defined as \cite{UhlmannFidelity}:
\begin{equation}
F(\sigma, \rho) \equiv \sup_{u' \in \mathcal{A}': u' {u'}^* =1} | \left<  \xi_{\sigma} | u' \xi_{\rho} \right> | . 
\end{equation}
Some of its properties in our setting are discussed in lem. \ref{lem:fid} below.

One of the two main theorems we would like to establish is:
\begin{theorem}[Faithful case]
\label{thm1}
Montonicity of relative entropy can be strengthened to
\begin{equation}
\label{streng}
 S_{\cA}(\rho| \sigma ) -  S_{\cB}(\rho| \sigma )\geq
 - 2 \int_{-\infty}^{\infty} \ln F( \rho, \rho \circ \iota \circ \alpha_\sigma^t ) p(t) \, \ud t ,
 \end{equation}
where we assume that $\rho,\sigma$ are normal, $\sigma$ is faithful and
where $\alpha_\sigma^t : \mathcal{A} \rightarrow \mathcal{B}$ is the rotated Petz map, defined as
\begin{equation}
\label{petzfaith}
\alpha_{\sigma}^t (a) =\varsigma^{\sigma,\mathcal{B}}_{-t} \left(  J_{\mathcal{B}} V_{\sigma}^*  J_\cA  \varsigma^{\sigma, \cA}_t(a)  J_\cA V_{\sigma} J_{\mathcal{B}}  \right).
\end{equation}
$p(t)$ is the normalized probability density defined by
\begin{equation}
p(t) = \frac{\pi}{\cosh(2\pi t) + 1}. 
\end{equation}
$\varsigma^{\sigma, \cA}_t$ resp. $\varsigma^{\sigma, \mathcal{B}}_t$ are the modular flows of $\sigma$
on $\cA$ resp. of $\sigma \circ \iota$ on $\cB$.
\end{theorem}

We may extend this theorem to the case where $\sigma$ is non-faithful. The basic idea is contained in the following lemma:
\begin{lemma}
\label{lem:proj}
Consider a sub-algebra $\iota(\mathcal{B}) \subset \mathcal{A}$, of a general v. Neumann algebra, and a normal state $\sigma$ with support
projectors $\pi^\cA(\sigma),\,\,\pi^\mathcal{B}(\sigma)$ and $\pi^{\cA'}(\sigma) \equiv J_\cA \pi^\cA(\sigma) J_\cA, \,\, \pi^{\mathcal{B}'}(\sigma) \equiv J_\mathcal{B} \pi^\mathcal{B}(\sigma) J_\mathcal{B}$. 
Then the following statements hold:
\begin{itemize}
\item[(i)] The projected sub-algebras, are ($\sigma$-finite) v. Neumann sub-algebras,
\begin{equation}
\iota_\pi(\mathcal{B}_\pi) \subset \mathcal{A}_\pi, 
\end{equation}
\begin{equation}
\mathcal{A}_\pi = \pi^\cA(\sigma) \mathcal{A} \pi^{\cA}(\sigma) \pi^{\cA'}(\sigma) \,, \qquad 
\mathcal{B}_\pi = \pi^{\mathcal{B}}(\sigma) \mathcal{B} \pi^{\mathcal{B}}(\sigma) \pi^{\mathcal{B}'}(\sigma) 
\end{equation}
acting respectively on $\mathscr{H}_\pi = \pi^\cA(\sigma) \pi^{\cA'}(\sigma) \mathscr{H}$ and  $\mathscr{K}_\pi = \pi^{\mathcal{B}}(\sigma) \pi^{\mathcal{B}'}(\sigma) \mathscr{K}$.
The projected inclusion is defined as:
\begin{equation}
\iota_\pi (b) \equiv  \Phi^{-1}_\cA \circ \iota \circ \Phi_{\mathcal{B}}(b) \, \qquad b \in \mathcal{B}_\pi,
\end{equation}
where we defined the (ultra weakly continuous) *-isomorphism of v. Neumann algebras 
\begin{subequations}
\begin{align}
\Phi_{\mathcal{B}} &: \mathcal{B}_\pi \rightarrow \pi^{\mathcal{B}}(\sigma) \mathcal{B} \pi^{\mathcal{B}}(\sigma)\, \quad {\rm via} \,\quad \Phi_{\mathcal{B}}( b \pi^{\mathcal{B}'}(\sigma)) = b \\
\label{defPhi}
\Phi_\cA &: \mathcal{A}_\pi \rightarrow \pi^\cA(\sigma) \mathcal{A} \pi^\cA(\sigma)\, \quad {\rm via} \,\quad \Phi_\cA( a \pi^{\cA'}(\sigma)) = a.
\end{align}
\end{subequations} 
The projected algebras are in a standard form. For example the standard form of $\mathcal{A}_\pi$ is $(\mathcal{A}_\pi, \mathscr{H}_\pi, J_\cA, \pi(\sigma) \pi'(\sigma)\mathscr{P}^\natural)$ where $J_\cA$ maps the subspace $\mathscr{H}_\pi$ to itself. 

\item[(ii)]  The relative entropy satisfies
\begin{equation}
S( \rho | \sigma) = S( \rho \circ \Phi | \sigma \circ \Phi ) \,, \qquad S( \rho \circ \iota | \sigma \circ \iota ) = S( \rho \circ \Phi \circ \iota_\pi | \sigma \circ \Phi\circ \iota_\pi ) 
\end{equation}
for all states such that $\pi^\cA(\rho) \leq \pi^\cA(\sigma)$, where $\Phi \equiv \Phi_\cA$. 
\item[(iii)] Consider a channel on the projected algebras: 
\begin{equation}
\alpha_\pi \, : \,\, \mathcal{A}_\pi \rightarrow \mathcal{B}_\pi
\end{equation}
We can construct a new cannel on the algebras of interest $\alpha \, : \, \mathcal{A} \rightarrow \mathcal{B}$ via:
\begin{equation}
 \alpha( a ) \equiv \Phi_{\mathcal{B}} \circ \alpha_\pi \circ \Phi^{-1}_\cA (\pi^\cA(\sigma) a \pi^\cA(\sigma))  
 + \sigma(a) ( 1 -\pi^{\mathcal{B}}(\sigma)).
\end{equation}
 Then for all $\rho \in \mathcal{A}_\star$ with $\pi^\cA(\rho) \leq \pi^\cA(\sigma)$ we have:
\begin{align}
 \rho( a) = \rho( \pi^\cA(\sigma) a \pi^\cA(\sigma) ) \quad \text{and} \quad  \rho \circ \iota \circ \alpha( a) =  \rho \circ \iota \circ \alpha(  \pi^\cA(\sigma) a \pi^\cA(\sigma)  ) \,, \qquad \forall a \in \mathcal{A} 
\end{align}
and
\begin{equation}
F(\rho, \rho \circ \iota \circ \alpha) = F(\rho \circ \Phi, \rho \circ \iota \circ \alpha \circ \Phi)  
= F(\rho\circ \Phi, \rho \circ \Phi \circ \iota_\pi \circ \alpha_\pi )
\end{equation}
Similarly:
\begin{equation}
\label{equalfunc}
\| \rho - \rho \circ \iota \circ \alpha \| = \| \rho\circ \Phi - \rho \circ \Phi \circ \iota_\pi \circ \alpha_\pi \|.
\end{equation}

\item[(iv)] The explicit form of the resulting Petz map coming from the inclusion $\iota_\pi(\mathcal{B}_\pi) \subset \mathcal{A}_\pi$ is:
\begin{equation}
\label{petznonfaith}
\alpha_{\sigma}^t (a)  \equiv \Phi_{\mathcal{B}}\left( \varsigma^{\sigma;\mathcal{B}}_{-t} \left(  J_{\mathcal{B}} (V_{\sigma}^{(\iota_\pi)})^*  J_\cA  \varsigma^{\sigma, \cA}_t(a)  J_\cA V_{\sigma}^{(\iota_\pi)} J_{\mathcal{B}}  \right) \right) + \sigma(a) (1-\pi^{\mathcal{B}}(\sigma)),
\end{equation}
where the embedding $V_{\sigma}^{(\iota_\pi)}$ is defined for the projected inclusion as
\begin{equation}
\label{embedproj}
V_\sigma^{(\iota_\pi)}  \left( b \big| \xi_{\sigma}^{\mathcal{B}} \big> \right)  = \iota_\pi( b ) \left| \xi_\sigma^\cA \right> \,, \qquad b\in \mathcal{B}_\pi ,
\end{equation}
and where $|\xi_\sigma^\cA\rangle$ and $|\xi_{\sigma}^{\mathcal{B}} \rangle$ are now cyclic and separating for $\mathcal{A}_\pi$ and $\mathcal{B}_\pi$ respectively. 

\end{itemize}
\end{lemma}

\begin{proof} 
The proof of this lemma uses standard properties of support projectors and is left to the reader.
\end{proof}

Note that the modular automorphism groups in \eqref{petznonfaith} can be understood as being associated to the non-cyclic and separating vector $|\xi_\sigma^\cA\rangle$ (resp. $|\xi_{\sigma}^{\mathcal{B}} \rangle$) for the original
algebra $\mathcal{A}$ (resp. $\mathcal{B}$), which are however
defined to project to zero away from the $\mathscr{H}_\pi$ (resp. $\mathscr{K}_\pi$) subspace. So, for example $ \varsigma^{\sigma}_{t=0}(a) = \pi^{\cA'}(\sigma) \pi^\cA(\sigma) a \pi^\cA(\sigma)$.
Similarly $V_\sigma^{(\iota_\pi)}$ can be understood in this way, as being defined on the subspaces $\mathscr{K}_\pi$ and projecting to zero away from this subspace via:
\begin{equation}
V_\sigma^{(\iota_\pi)}  \left( \pi^{\mathcal{B}}(\sigma) b \big| \xi_{\sigma}^{\mathcal{B}} \big> + \left| \zeta \right> \right)  = \iota( \pi^{\mathcal{B}}(\sigma))  \iota(b) \left| \xi_\sigma^\cA \right> \,, \quad b\in \mathcal{B} \,, \quad \left| \zeta \right> \in (1-\pi^{\mathcal{B}'}(\sigma)\pi^{\mathcal{B}}(\sigma)) \mathscr{K}
\end{equation}

An obvious corollary is:
\begin{corollary}[Theorem 1 in the non-faithful case]
Theorem \ref{thm1} continues to hold when $\sigma$ is non-faithful but still $\pi^\cA(\rho) \leq \pi^\cA(\sigma)$. The recovery map is now given by \eqref{petznonfaith}.
\end{corollary}

From this result one can characterize approximately sufficiency using relative entropy:

\begin{theorem}
\label{thm2}
\label{thm2p}
Consider a set of normal states $\sS$ on a general v. Neumann algebra 
$\mathcal{A}$ with a subalgebra $\mathcal{B}$.  
If $\sS$ contains a state $\sigma$ such that for all $\rho \in \sS$ the following condition holds:
\begin{equation}
S(\rho | \sigma) < \infty  \quad \text{and} \quad S_{\cA}(\rho | \sigma) - S_{\mathcal B}( \sigma | \rho ) \leq \delta,
\end{equation}
then there exists a universal recovery channel $\alpha_{\sS}$ such that $\mathcal{A} \subset \mathcal{B}$ is $\epsilon$-approximately sufficient for $\sS$. 
(Here $\delta = -\ln(1-\epsilon^2/4)$.)

The explicit form of the recovery map is:
\begin{equation}
\label{recover}
\alpha_{\sS} : \cA \owns a \mapsto \int_{-\infty}^{\infty} \alpha_\sigma^t(a)  \, p(t) \ud t \in \mathcal{B},
\end{equation}
where $\alpha_\sigma^t$ was given in \eqref{petznonfaith}.
We can make sense of the later integral as a  Lebesgue integral of a weakly measurable function with values in $\mathcal{B}$, thought of as a Banach space.
\end{theorem}

\begin{remark}
Less powerful antecedents of thms. \ref{thm1}, \ref{thm2} can be found in \cite{Berta,Carlen,Jencova,JencovaPetz,Sutter,Wilde}.
\end{remark}

An example of a set of states that satisfy the assumptions in thm.~\ref{thm2p} is simply $\sS = \sR_{\delta}^{(\sigma)}$ \eqref{Rsigma} for any state $\sigma$. If we were to additionally assume that $\mathcal{A}$ is $\sigma$-finite then we could also pick $\sS$ to be any closed convex set of states such that
\begin{equation}\label{eq:cond}
\rho_{1,2} \in \sS\,,\quad \pi^\cA(\rho_1) \leq \pi^\cA(\rho_2) \implies S_{\cA}(\rho_1|\rho_2)< \infty  \quad \text{and} \quad   S_{\cA}(\rho_1 | \rho_2) - S_{\mathcal B}( \rho_1 | \rho_2 ) \leq \delta.
\end{equation}
To see this, note that the $\sigma$-finite condition imposes that all families of mutual orthogonal projectors in $\mathcal{A}$ are at most countable.  
This is satisfied for v. Neumann algebras that act on a separable Hilbert space, and is equivalent to the assumption that there is a faithful state in $\mathcal{A}_\star$. 
Then \eqref{eq:cond} is sufficient for finding a $\sigma$ that works with thm.~\ref{thm2p} due to the following basic result:

\begin{lemma}\label{lem:2}
Given a closed convex subset of normal states $\sS \subset \mathcal{A}_\star$ for a $\sigma$-finite v. Neumann algebra $\mathcal{A}$ then we can always find a $\sigma \in \sS$ such that:
\begin{equation}
\label{allincl}
\pi^\cA(\rho)  \leq \pi^\cA(\sigma) \,, \qquad \forall \rho \in \sS
\end{equation}
\end{lemma}
\begin{proof}
Given in app. \ref{lemma2}.
\end{proof}

\section{Proof of main theorems}

Our eventual goal in this section is to prove our main results, thm.s \ref{thm1} and \ref{thm2}. As discussed above, without loss of generality we can take $\sigma$ to be faithful and so we will assume this from now on. 

The proof is divided into several 
steps. In subsec. \ref{subsec:5.1}, we first fix some notation and recall basic facts about the vectors that we are dealing with. 
In subsec. \ref{subsec:5.2}, we introduce the non-commutative $L_p$-space by Araki and Masuda \cite{AM} and explain its -- in principle well-known -- relation to the fidelity. We make certain minor modifications to the standard setup and prove a simple but important intermediate result which we call a ``first law'', thm. \ref{firstlaw}. In subsec. \ref{subsec:5.3}, we motivate the definition of certain interpolating vectors that will be of main interest in the following subsections and in subsec. \ref{subsec:5.4} we prove some of their basic properties. 
Subsec. \ref{subsec:5.5}
is the most technical section. It introduces certain regularized (``filtered'') versions of our interpolating vectors and their properties. Our definition of filtered vectors involves a certain cutoff, $P$, that is defined in terms of relative modular operators. A quite general result of independent interest is that the relative entropy behaves continuously 
as this cutoff is removed, thm. \ref{lem:finites}. Armed with this technology, we can then complete the proofs in 
subsec. \ref{subsec:5.6} using an interpolation result for Araki-Masuda $L_p$ spaces, lem. \ref{lem:hirsch}.

\subsection{Isometries $V_\psi$ for general states, notation}
\label{subsec:5.1}

Since the two states $\sigma,\rho$ play a central role in thm.~\ref{thm1} we will use a special notation for the vectors that represent these states in their respective natural cones:
\begin{equation}
\label{eq:notation}
\left| \eta_\cA \right> \equiv \left| \xi_\sigma^\cA \right> \,, \quad \left| \eta_{\mathcal{B}} \right> \equiv \big| \xi_{\sigma}^{\mathcal{B}} \big>\,, \qquad 
\left| \psi_\cA \right> \equiv \left| \xi_\rho^\cA \right> \,, \quad \left| \psi_{\mathcal{B}} \right> \equiv \big| \xi_{\rho}^{\mathcal{B}} \big>\,
\end{equation}
where $|\eta_\cA \rangle \in \mathscr{H}\,\,( |\eta_{\mathcal{B}} \rangle \in \mathscr{K})$ are cyclic and separating for $\mathcal{A} \,\,(\mathcal{B})$. 

We will also choose to label various objects, such as support projectors, and the modular operators discussed below, for the most part with the vector rather than the linear functional as we did in Section~\ref{sec:above}.  This will be convenient since we will occasionally have to work with vectors that do not necessarily live in the natural cone.  
For example, given a $\ket{\chi} \in \mathscr{H}$ we define:
\begin{equation}
\pi^{\cA'}(\chi) \equiv \pi^{\cA'}( \omega_\chi' ) \,, \qquad \pi^\cA(\chi) \equiv \pi^\cA( \omega_\chi )
\end{equation}
where $\omega_\chi' \in \mathcal{A}'_\star$ is the induced linear functional of $\ket{\chi} \in \mathscr{H}$ on the commutant. 
For vectors $\ket{\xi}$ in the natural cone we have a symmetry between the support projectors $\pi^{\cA}(\xi)= J_\cA \pi^\cA(\xi)J_\cA$. We use similar notation for objects associated to the algebras $\mathcal{B}$. When the only algebra 
in question is $\cA$, we write 
\ben
\pi(\chi) \equiv \pi^{\cA}(\chi), \quad \pi'(\chi) \equiv \pi^{\cA'}(\chi).
\een
We have already recalled that a general vector $|\chi \rangle \in \sH$ is related
to a unique vector in the natural cone inducing the same linear functional on $\mathcal{A}$. More precisely, 
there is a partial isometry in $v'_\chi \in \mathcal{A}'$ such that 
\begin{equation}
\label{genvec}
\left| \chi \right> = {v'_{\chi}}^* \left| \xi_\chi \right>\, , \qquad   v'_{\chi} {v'_{\chi}}^* =\pi^{\cA'}(\chi) 
 \,, \qquad {v'_{\chi}}^* v'_{\chi} = \pi^{\cA'}(\chi)
\end{equation}
Now consider a vector $|\psi_\cA\rangle = |\xi^\cA_\psi\rangle \in \mathscr{P}^\natural_\cA$  and define a corresponding vector in $\mathscr{K}$ using $
|\psi_{\mathcal{B}}\rangle \equiv \xi^{\mathcal{B}}_{\psi} \in \mathscr{P}^\natural_{\mathcal{B}}$. 
The  vector $V_\eta \left| \psi_{\mathcal{B}} \right> \in \mathscr{H}$ induces the same linear functional on $\iota(\mathcal{B})$ as $\left| \psi_\cA \right>$,
where we use exchangeably the notation $V_\eta = V_\sigma$ for the embedding \eqref{embed}. Thus there exists a partial isometry $u_{\psi;\eta}'$in $\iota(\mathcal{B})'$, with implied initial and final support, relating the two vectors 
\begin{equation}
\label{Vetaother}
V_\eta \left| \psi_{\mathcal{B}} \right> ={u_{\psi,\eta}' }^* \left| \psi_\cA \right>.
\end{equation}
Combining this with \eqref{embedcomm} we have for $b \in \mathcal{B}$
\begin{equation}
V_\eta b \left| \psi_{\mathcal{B}} \right>  = \iota(b)  {u_{\psi,\eta}' }^*  \left| \psi_\cA \right>.
\end{equation}
Since this notation is cumbersome we will simply define a new isometry 
$V_\psi :  \mathscr{K} \rightarrow \mathscr{H}$ that is defined with reference to $\ket{\psi}$
\begin{equation}
V_\psi \equiv u_{\psi,\eta}' V_\eta\, .  
\end{equation}
It will also be convenient to have $V_\chi$ defined for states $|\chi\rangle \in \sH$ that are not necessarily 
in the natural cone. In that case, we extend this definition further:
\begin{equation}
\label{Vchi}
\quad V_\chi \equiv {v'_\chi}^* u_{\xi_\chi;\xi_\eta}' V_\eta   \, .
\end{equation}
These satisfy
\begin{equation}
\label{Vpsiprop}
V_\chi b \left| \xi^{\mathcal B}_\chi \right> = \iota(b) \left| \chi \right> .
\end{equation}

\subsection{$L_p$ spaces, fidelity and relative entropy}
\label{subsec:5.2}

In this part we introduce various quantum information measures that will be useful to characterize sufficiency. We have already seen the importance of relative entropy and the fidelity. What we need are quantities interpolating between them. These will be provided by the non-commutative $L_p$ norm associated with a v. Neumann algebra, with reference to a state/vector. There exist different definitions of such norms/spaces in the literature; here we basically follow the version by Araki and Masuda \cite{AM}, suitably generalized to non-faithful states. Such a generalization was considered up to a certain extent by \cite{Berta2}, see also \cite{JencovaLp1} for related work.

\begin{definition} \cite{AM} Let $\mathcal M$ be a v. Neumann algebra in standard form acting on a Hilbert space $\sH$. For $1 \leq p \leq 2$ the Araki-Masuda $L_p({\mathcal M}, \psi)$ norms, with reference to a fixed vector $\ket{\psi} \in \mathscr{H}$, are defined by\footnote{The Araki-Masuda norms were originally defined assuming a faithful normal reference state. 
For the most part we will only ever need the definition of the norm \eqref{AMdef} for vectors in the Hilbert space, along with some simple consequences of this variational formula. Thus we will not need the full machinery developed by \cite{AM}, except at some crucial steps in the interpolation argument below that we will highlight. When this is the case we will apply their results for a \emph{faithful state} and prove that one can extrapolate to the case at hand. 
}:
\begin{equation}
\label{AMdef}
\left\| \zeta \right\|_{p,\psi}^{\mathcal{M}}
= \inf_{ \chi \in \mathscr{H}:  \| \chi \| = 1, \pi^{\mathcal{M}}(\chi) \geq \pi^{\mathcal{M}}(\zeta) }
\| \Delta_{\chi,\psi}^{1/2-1/p} \zeta  \| 
\end{equation}
where the definition above only depends on the functional $\omega_\psi$ but not the choice of vector representative, $|\psi\rangle$. 
\end{definition}

\begin{remark}
1)
The norm is always finite for this range of $p$. We will use the $L_p$ norms mostly for the commutant algebra $\cA'$ of $\cA$. Then,
\begin{equation}
\label{projequal}
 \left\| \zeta \right\|_{p,\psi}^{\mathcal{A}'}=  \left\| \pi^\cA(\psi) \zeta\right\|_{p,\psi}^{\mathcal{A}'},
\end{equation}
due to the (possibily) restricted support of the relative modular operator. 

2) For $1\ge\alpha\ge1/2$, the quantity 
$\frac{1}{\alpha-1} \ln\| \eta \|_{2\alpha,\psi}^{2\alpha}$ is sometimes called the ``sandwiched Renyi entropy'' (between $\ket{\eta}, \ket{\psi}$). 
It is in general different from the ``usual'' Renyi-Petz entropy, $\frac{1}{\alpha-1} \ln \langle \psi | \Delta_{\psi, \eta}^{1-\alpha} \psi \rangle$. 
Both quantities, as well as the $L_p$ norms, can be defined or more general values of the parameters but are not needed here.
\end{remark}

When $p=2$, the $L_p$ norm becomes the projected Hilbert space norm:
\begin{equation}
\label{eq:h2}
\left\|  \zeta  \right\|_{2,\psi}^{\mathcal{A}'} = \left\| \pi^\cA(\psi) \zeta  \right\|.
\end{equation}
Taking a derivative at $p=2$ will give the relative entropy comparing $\zeta$ with $\psi$ 
as linear functionals on $\mathcal{A}$, see below. 

At $p=1$ we have the following lemma:
\begin{lemma}
\label{lem:fid}
\begin{enumerate}
\item
At $p=1$ the Araki-Masuda norm \eqref{AMdef} relative to $\cA'$ becomes the fidelity
\begin{equation}
\label{uhl1}
\left\|  \phi  \right\|_{1,\psi} = F(\omega_\psi,\omega_\phi) \equiv \sup_{u' \in \mathcal{A}': (u')^* u' = 1} |\left< \psi \right| u' \left| \phi \right>|,
\end{equation}
where $\omega_\phi,\omega_\psi \in \mathcal{A}$ are the induced linear functionals for $\ket{\phi},\ket{\psi}$, respectively. 

\item
The fidelity may also be written as
\begin{equation}
\label{uhl2}
F(\omega_\psi,\omega_\phi) = \sup_{x' \in \mathcal{A}': \| x'\| \leq 1} |\left< \psi \right| x' \left| \phi \right>|.
\end{equation}

\item 
It is related to the linear functional norm (Fuchs-van-der-Graff inequalities) by
\begin{equation}
1 -  F(\omega_\psi,\omega_\phi)  \leq  \frac{1}{2} \| \omega_\psi - \omega_\phi \| \leq   \sqrt{  1 - F(\omega_\psi,\omega_\phi) ^2 }.
\end{equation}
\end{enumerate}
\end{lemma}
\begin{proof}
While these results are standard, we include the proof in the app.~\ref{app:lem:fid} because we also treat the non-faithful case for the generalized Araki-Masuda norm in \eqref{AMdef} which has not explicitly appeared elsewhere as far as we are aware. Note that an argument conditional on other -- unproven in the non-faithful case -- properties of Araki-Masuda norms was given in \cite{Berta2}.
\end{proof}

We will also need the following result that is potentially of independent interest. 
\begin{theorem}[First Law for Renyi Relative Entropy]
\label{firstlaw}
Consider a one parameter family of vectors $\ket{\zeta_\lambda} \in \mathscr{H}$ for $\lambda \geq 0$, which are normalized $ \|  \zeta_\lambda  \| =1$ and satisfy
\begin{equation}
\label{vanish}
\lim_{\lambda \rightarrow 0^+} \frac{\left\| \zeta_\lambda -\psi  \right\|^2}{\lambda} = 0 ,
\end{equation}
where $|\psi \rangle = |\zeta_0\rangle$. 
Then:
\begin{itemize}
\item[1)] The Petz-Renyi relative entropy satisfies:
\begin{equation}
\label{firstpetz}
\lim_{\lambda \rightarrow 0^+}  \frac{1}{\lambda} \ln \left< \zeta_\lambda \right|  \Delta_{\psi, \zeta_\lambda}^{x(\lambda)} \left|  \zeta_\lambda \right>
= 0  \,, \qquad 0 \leq x(\lambda) \leq 1 -\epsilon,
\end{equation}
where $\epsilon > 0$ and there is no other constraint on $x(\lambda)$. 
\item[2)] The sandwiched Renyi relative entropy satisfies:
\begin{equation}
\label{firstsand}
\lim_{\lambda \rightarrow 0^+}  \frac{1}{\lambda} \ln \left\| \zeta_\lambda \right\|^{\mathcal{A}'}_{p(\lambda),\psi}
= 0  \,, \qquad 1 \leq p(\lambda) \leq 2,
\end{equation}
with no other constraint on how the function $p(\lambda)$ behaves under the limit. 
\end{itemize}
\end{theorem}

In order to prove this, we first prove the following lemma:

\begin{lemma} 
\label{lem:est} Given two normalized vectors $\ket{\psi}, \ket{\zeta} \in \sH$, we have:
\begin{itemize}
\item[1)] For compact subsets $K$ of the complex strip $\{0 \leq {\rm Re} z < 1\}$, there exists a constant $C_K$ such that:
\begin{equation}
0 \leq {\rm Re}  \left(  1-  \left< \zeta \right| \Delta_{\psi,\zeta}^z \left| \zeta \right> \right) \leq C_K \left\| \zeta-\psi \right\|^2
\end{equation}
for all $z\in K$. $C_K$ is independent of $\ket{\psi}, \ket{\zeta}$. 

\item[2)] We also have for $1 \leq p \leq 2$:
\begin{equation}
\label{petzAM}
0 \leq   1 - \left< \zeta \right| \Delta_{\psi,\zeta}^{2/p-1} \left| \zeta \right> \leq 1 - \left( \left\|   \zeta  \right\|^{\mathcal{A}'}_{p,\psi} \right)^2 \leq 1 - \left(\left< \zeta \right| \Delta_{\psi,\zeta}^{1-p/2} \left| \zeta \right>  \right)^{2/p},
\end{equation}
and we have the elementary bound:
\begin{equation}
1 - \left(\left< \zeta \right| \Delta_{\psi,\zeta}^{1-p/2} \left| \zeta \right>  \right)^{2/p} \leq \frac{2}{p} \left( 1 - \left< \zeta \right| \Delta_{\psi,\zeta}^{1-p/2} \left| \zeta \right>  \right) .
\end{equation}
\end{itemize}
\end{lemma}
\begin{proof}
(1) This is demonstrated by an application of Harnack's inequality (see e.g. \cite{evans}, sec. 2, thm. 11) which applies to any $h(z)$ that is harmonic and non negative in some connected open set $O$: for all compact subsets $K \subset O$ there exists a constant $1 \leq C(K,O) < \infty$ such that:
\begin{equation}
h(z) \leq C(K,O) h(w)\,, \qquad \forall z,w \in K,
\end{equation} 
where notably this constant is independent of the particular $h$ satisfying the assumptions. 

We work with the real part of two holomorphic functions in two strips:
\begin{subequations}
\begin{align}
h_1(z) &= {\rm Re} (1 - \left< \psi \right| \Delta_{\psi,\zeta}^z \left| \zeta \right>) \,, \,\qquad  O_1 = \{ z \in \mathbb{C}: -1/2 < {\rm Re }(z)  < 1/2 \} \\
 h_2(z) &= {\rm Re} ( 1 - \left< \zeta \right| \Delta_{\psi,\zeta}^z \left| \zeta \right>) \,, \,\qquad  O_2 = \{ z \in \mathbb{C}: 0 < {\rm Re }(z)  < 1 \}.
\end{align}
\end{subequations}
These functions are continuous on the closure of the above strips 
and they are non-negative since for normalized vectors $|\left< \psi \right| \Delta_{\psi,\zeta}^z \left| \zeta \right>|,|\left< \zeta \right| \Delta_{\psi,\zeta}^z \left| \zeta \right>| \leq 1$ by an easy application of the Hadamard three lines theorem -- these facts are standard results of Tomita-Takesaki
theory for the relative modular operators. There is no need for any of the vectors to be in the natural cone.

We can thus apply Harnack's inequality. Using the fact that:
\begin{equation}
h_1(0) = \frac{1}{2} \| \psi- \zeta\|^2,
\end{equation}
and picking the compact subset $K_1 \subset O_1$ with $0 \in K_1$ we have:
\begin{equation}
0\leq h_1(z) \leq \frac{1}{2} C(K_1,O_1)  \| \psi-\zeta \|^2 \qquad \forall \,\, z \in K_1.
\end{equation}
We have to relate this to $h_2(z)$ which is what we are most concerned with. We can relate the two functions using the Cauchy-Schwarz inequality
where the two defining strips overlap, $0\leq {\rm Re} z \leq 1/2$:
\begin{align} \nonumber
& {\rm Re} \left(  - \left< \zeta \right|  \Delta_{\psi,\zeta}^z \left| \zeta \right> +\left< \psi \right| \Delta_{\psi,\zeta}^z \left| \zeta \right> + \left< \zeta \right| \left. \psi \right> -1 \right) \leq 
\left| \left( \left| \zeta \right> - \left| \psi \right>,   \Delta_{\psi,\zeta}^z \left| \zeta \right>  - \left| \psi \right> \right) \right|\\ 
 & \qquad \qquad \leq \left\| \zeta-\psi\right\| \left(  {\rm Re} \left( 1 +  \left< \zeta \right|  \Delta_{\psi,\zeta}^{z+\bar z} \left| \zeta \right>   - 2 \left< \psi \right|  \Delta_{\psi,\zeta}^z \left| \zeta \right> \right)  \right)^{1/2}   \nonumber \\ &
 \qquad \qquad \leq \left\| \zeta-\psi \right\| \left(  {\rm Re} \left( 2   - 2 \left< \psi \right|  \Delta_{\psi,\zeta}^z \left| \zeta \right> \right)  \right)^{1/2}  ,
\end{align} 
which translates to:
\begin{equation}
h_2(z) 
\leq \frac{1}{2} \left( \| \zeta-\psi \| + \sqrt{ 2 h_1(z)} \right)^2,
\end{equation}
so that
\begin{equation}
\label{bdL1}
0\leq h_2(z) \leq \frac{1}{2} \left( 1 + \sqrt{C(K_1,O_1)} \right)^2   \| \psi-\zeta \|^2 \qquad \forall \,\, z \in K_1 \cap \overline{O}_2.
\end{equation}

We can split the compact set 
$K$ in the statement of the lemma into two compact pieces $K_1 = K \cap \{ z\in \mathbb{C}: 0 \leq {\rm Re} z \leq 1/4 \}$
and $K_2 = K \cap \{ z\in \mathbb{C}: 1/4 \leq {\rm Re} z \leq 1 \}$.
These satisfy $K_i \subset O_i \,\,: \, i=1,2$. Repeatedly applying Harnack's inequality as above gives the following upper bound for $C_K$:
\begin{equation}
  \frac{1}{2} {\rm max} \left\{ \left( 1 + \sqrt{C(\{0\} \cup K_1,O_1)} \right)^2 ,\,\, C(\{ \tfrac{1}{4} \} \cup K_2,O_2) \left( 1 + \sqrt{C(\{0, \tfrac{1}{4}\},O_1)} \right)^2 \right\},
\end{equation}
where it was necessary to add the points $\{ 0,1/4 \}$ since they may not have been in the original $K$.

(2) This result is basically the well-known Araki-Lieb-Thirring inequality \cite{Araki3}, for a proof in the v. Neumann algebra setting see \cite{Berta2}, thm. 12, for $L_p$ norms based on a not necessarily cyclic and separating vector $\ket{\psi}$.
\end{proof}

\begin{proof}[Proof of Theorem~\ref{firstlaw}] (1) is a consequence of lem.~\ref{lem:est} (1): We can take $K= [0,1-\epsilon]$ which satisfies the assumptions
of this lemma so:
\begin{equation}
0 \leq  \lim_{\lambda \rightarrow 0^+} \frac{\left(1-\left< \zeta_\lambda \right|  \Delta_{\psi, \zeta_\lambda}^{x(\lambda)} \left| \zeta_\lambda \right> \right)}{\lambda} 
\leq \lim_{\lambda \rightarrow 0^+}  C_K \frac{\|  \zeta_\lambda -\psi \|^2}{\lambda} =0.
\end{equation}
Then using differentiability of $\ln(x)$ at $x=1$ and the chain rule we show \eqref{firstpetz}.

(2) Here we need lem.~\ref{lem:est} (1) with $K=[0,1/2]$. Applying lem.~\ref{lem:est} (2):
\begin{align}
0 \leq  \lim_{\lambda \rightarrow 0^+} \frac{\left(1-\left\|  \zeta_\lambda  \right\|^{\mathcal{A}'}_{p(\lambda),\psi} \right)}{\lambda} 
& \leq \lim_{\lambda \rightarrow 0^+} \frac{2}{p(\lambda)} \left( 1 - \left< \zeta_\lambda \right| \Delta_{\psi,\zeta_\lambda}^{1-p(\lambda)/2} \left| \zeta_\lambda \right>  \right)  \nonumber \\ & \leq
\lim_{\lambda \rightarrow 0^+}  2 C_K \frac{\| \zeta_\lambda -\psi \|^2}{\lambda} =0.
\end{align}
Again differentiability of $\ln(x)$ at $x=1$ and the chain rule gives \eqref{firstsand}.
\end{proof}

\subsection{Exact recoverability/sufficiency} 
\label{subsec:5.3}

This section is meant as an informal summary of some of the results given in \cite{Petz3,Petz4}, defining the exact notion of recoverability or sufficiency. 
We will focus only on the properties associated to sufficiency that we make contact with in this paper, and we will also treat only the case of faithful linear functionals and drop all support projectors here. 

By definition, the quantum channel $\iota: \mathcal{B} \to \cA$ is exactly reversible for at least two fixed states $\rho, \sigma$ if there exists a recovery channel $\alpha:\cA \to \mathcal{B}$ such that:
\begin{equation}
\rho \circ \iota \circ \alpha (a) = \rho(a)\,, \qquad \forall a \in \cA,
\end{equation}
and similarly for $\sigma$.
Since the relative entropy is monotonous \cite{Uhlmann1977} under both $\alpha, \iota$, we must have 
$S_\cA(\rho,\sigma)=S_\cB(\rho,\sigma)$, see \eqref{eq:SB} for our notation. 
Representing $\sigma, \rho$ by vectors in the natural cone as in \eqref{eq:notation} and using 
a standard integral representation of the relative entropy based on the spectral theorem and the elementary identity $(x,y>0)$
\ben
\ln y - \ln x = \int_0^\infty \left( \frac{1}{x+\beta} - \frac{1}{y+\beta} \right) \ud \beta \ , 
\een
we get that
\begin{align}
S_{\cA}(\rho | \sigma) -S_{\mathcal B}( \rho | \sigma)   
= \int_0^\infty  \left< \psi_{\mathcal{B}} \right| \left( V_\psi^* \frac{1}{\beta + \Delta_{\eta_\cA,\psi_\cA}}  V_\psi- \frac{1}{\beta + \Delta_{\eta_{\mathcal{B}},\psi_{\mathcal{B}}}} \right) \left| \psi_{\mathcal{B}} \right> \ud \beta
\end{align}
vanishes. Known properties of the modular operators imply that the integrand is positive \cite{Petz3,Petz4,Petz1993}. 
Therefore, 
\begin{equation}
V_\psi^* \frac{1}{\beta + \Delta_{\eta_\cA,\psi_\cA}}  \left| \psi_\cA  \right> =  \frac{1}{\beta + \Delta_{\eta_{\mathcal{B}},\psi_{\mathcal{B}}}} \left| \psi_{\mathcal{B}}  \right>
\end{equation}
for all $\beta > 0$,
which can be integrated against a specific kernel that we will not write to arrive at a statement about the 
relative modular flow:
\begin{equation}
V_\psi^* \Delta_{\eta_\cA,\psi_\cA}^{it} \left| \psi_\cA \right>  =  \Delta_{\eta_{\mathcal{B}},\psi_{\mathcal{B}}}^{it} \left| \psi_{\mathcal{B}}  \right>
\, \quad \implies \, \quad \left| \psi_\cA \right> = \Delta_{\eta_\cA,\psi_\cA}^{-it}   V_\psi \Delta_{\eta_{\mathcal{B}},\psi_{\mathcal{B}}}^{it} \left| \psi_{\mathcal{B}}  \right>.
\end{equation}
Further manipulations give a derivation that the Petz map is a perfect recovery channel, although we will not go through this. Here we simply note that it is a reasonable guess at this point that for the approximate version of recoverability, one must require that  $\left| \psi_\cA \right>$ must be close to $\Delta_{\eta_\cA,\psi_\cA}^{-it}   V_\psi \Delta_{\eta_{\mathcal{B}},\psi_{\mathcal{B}}}^{it} \left| \psi_{\mathcal{B}}  \right>$ in some metric.
We will use the non-commutative Araki-Masuda $L_p$ norms to provide such a metric.  

\subsection{Interpolating vector}
\label{subsec:5.4}

Motivated by the above discussion we consider the following vector in $\mathscr{H}$:
\begin{equation}
\label{ourvec}
\left| \Gamma_\psi(z) \right> = \Delta_{\eta_\cA,\psi_\cA}^z   V_\psi  \Delta_{\eta_{\mathcal{B}},\psi_{\mathcal{B}}}^{-z} \left| \psi_{\mathcal{B}} \right> , 
\end{equation}
defined at first for purely imaginary $z$, and assuming at first that $|\psi\rangle, |\eta\rangle$ are in the natural cone 
(of $\cA$), see \eqref{eq:notation} for our notation.

\begin{remark}
The vector defined here is similar in spirit but does not quite coincide with the interpolating vector considered by 
\cite{Junge}. It seems possible to consider other vectors instead, and we briefly comment on this in app. \ref{alterapproach}.
\end{remark}

Our first result will be an analytic continuation of the vector \eqref{ourvec} into a strip: 
\begin{theorem}
\label{thm:gamma}
\begin{enumerate}
\item
There is a vector-valued function $\left| \Gamma_\psi(z) \right>$ that is holomorphic in the strip $\bS_{1/2}=\{0 < {\rm Re }(z) < 1/2\}$, weakly continuous in the closure of the strip and has the following explicit form at the top and bottom edges:
\begin{subequations}
\begin{align}
\label{Gt}
\left| \Gamma_\psi(1/2 + i t) \right>  &=   \Delta_{\eta_\cA,\psi_\cA}^{it}  J_\cA    V_\eta J_{\mathcal{B}} \Delta_{\eta_{\mathcal{B}},\psi_{\mathcal{B}}}^{-it}   \left| \psi_{\mathcal{B}} \right> \, \, \\
\label{Gb}
\left| \Gamma_\psi(i t) \right> & =   \Delta_{\eta_\cA,\psi_\cA}^{it}  V_\psi \Delta_{\eta_{\mathcal{B}},\psi_{\mathcal{B}}}^{-it} \left| \psi_{\mathcal{B}} \right> .
\end{align}
\end{subequations}
The norm of the vector $|\Gamma_\psi(z)\rangle$ is bounded by $1$ everywhere in the closure of $\bS_{1/2}$, and $\big| \Gamma_\psi(0)\big> = \left| \psi_\cA \right>$. 

\item
On the top edge of the strip $\bS_{1/2}$ this vector induces the the following state on $\mathcal{A}$:
\begin{equation}
\label{inducepetz}
\left( \left| \Gamma_\psi(1/2 + i t) \right> , a_+ \left| \Gamma_\psi(1/2 + i t) \right>  \right) \leq
\left< \psi \right| \iota( \alpha_\eta^t(a_+)) \left| \psi \right> = \omega_\psi \circ \iota \circ \alpha_\eta^t(a_+) ,
\end{equation}
where $a_+$ is any non-negative self-adjoint element  in $\mathcal{A}$,
and where $\alpha_\eta^t$ is the rotated Petz map \eqref{petzfaith} for the state $\sigma$ induced by $|\eta\rangle$.
\end{enumerate}
\end{theorem}

\begin{remark}\label{rem:1}
1) A variant of this theorem holds when $|\psi\rangle$ is replaced by a unit vector $|\chi\rangle$ that is not necessarily 
in the natural cone. In this case, we should define
\begin{equation}
\label{comp}
\left| \Gamma_\chi(z) \right> = v'_{\chi} \left| \Gamma_{\xi_\chi} (z) \right>
\end{equation}
with $v'_\chi$ as in \eqref{genvec}. The limiting values \eqref{Gb}, \eqref{Gt} 
at the boundaries of the strip are then readily computed using
\eqref{Vetaother}. In particular, \eqref{Gb} takes the same form as before 
as seen using \eqref{Vchi}, \eqref{Vpsiprop}, which 
also implies $\big| \Gamma_\chi(0)\big> = \left| \chi \right>$. \eqref{inducepetz} follows from \eqref{eq:change}.

2) The proof shows that we would have equality in (2) if $\pi^{\cA'}(\psi) = 1$, i.e. if $|\psi\rangle$ is cyclic for $\cA$. 
\end{remark}

\begin{proof} 
Let us use in this proof the shorthands $\Delta_{\eta_{\mathcal{B}},\psi_{\mathcal{B}}}
=\Delta_{\eta, \psi;\mathcal{B}}$ and $\Delta_{\eta_{\mathcal{A}},\psi_{\mathcal{A}}}
=\Delta_{\eta, \psi;\mathcal{A}}$.

(1) Given an $a'\in \mathcal{A}'$, consider the  function:
\begin{equation}
g(z) = \left( \Delta^{\bar z-1/2}_{\eta,\psi;\cA} a' \left| \eta_\cA \right>,
J_\cA V_\eta J_{\mathcal{B}} \Delta_{\eta, \psi;\mathcal{B}}^{-z+1/2} \left| \psi_\mathcal{B} \right> \right) ,
\end{equation}
which using Tomita-Takesaki theory is analytic in the strip $\bS_{1/2}$,  continuous in the closure, and bounded by:
\begin{equation}
\label{weakbound}
|g(z)| \leq \max_{0\leq \theta \leq 1/2} \| (\Delta_{\eta,\psi;\cA}')^\theta a' \eta_\cA \| \,  \| (\Delta_{\eta, \psi;\mathcal{B}} )^\theta \psi_{\mathcal{B}} \|,
\end{equation}
where $\theta = {\rm Re}(1/2-z)$. The maximum is achieved by continuity and compactness of the interval. This bound is however not uniform over vectors $a' \left| \eta_\cA \right> \in \mathscr{H}$ with norm $1$. For this, we need to use the Phragmen-Lindel\" off theorem. Our function has the following form at the edges of the strip $(t \in \mathbb{R})$:
\begin{subequations}
\begin{align}
\label{geasy}
g(1/2+it) =& \left( a' \left| \eta_\cA \right>, \left| \Gamma_\psi(1/2+ i t) \right> \right) \\
g(it) =& \left( a' \left| \eta_\cA \right>, \left| \Gamma_\psi(i t) \right> \right) ,
\label{ghard}
\end{align}
\end{subequations}
where we made use of the expressions/definitions in \eqref{Gt} and \eqref{Gb} respectively. The first equation above is rather trivial but the second equation requires some lines of algebra:
\begin{align}
\label{manip}
g(it) & = \left( \Delta^{-it-1/2}_{\eta,\psi;\cA} a' \left| \eta_\cA\right>,
J_\cA V_\eta \Delta_{\psi, \eta;\mathcal{B}}^{-it} \pi^{\mathcal{B}}(\psi)  \left| \eta_\mathcal{B} \right> \right) \nonumber \\
& =  \left( \Delta^{-it-1/2}_{\eta,\psi;\cA} a' \left| \eta_\cA \right>,
J_\cA V_\eta b \left| \eta_\mathcal{B} \right> \right) \,, \qquad b = \Delta_{\psi, \eta;\mathcal{B}}^{-it} \pi^{\mathcal{B}}(\psi)   \Delta_{\eta;\mathcal{B}}^{it}
\in \mathcal{B} \nonumber \\
& =  \left( \Delta^{-it-1/2}_{\eta,\psi;\cA} a' \left| \eta_\cA\right>,
J_\cA  \iota(b) \left| \eta_\cA \right> \right) = 
\left( \Delta^{-it}_{\eta,\psi;\cA} a' \left| \eta_\cA\right>,  \iota(b)^* \left| \psi_\cA \right> \right) \nonumber \\
&
 = \left( \Delta^{-it}_{\eta,\psi;\cA} a' \left| \eta_\cA\right>,    V_\psi b^* \left| \psi_{\mathcal{B}} \right>\right)   = 
  \left( \Delta^{-it}_{\eta,\psi;\cA} a' \left| \eta_\cA\right>,    V_\psi  \Delta_{\eta, \psi;\mathcal{B}}^{-it} \left| \psi_{\mathcal{B}} \right>\right) ,
\end{align}
where in the first line we used \eqref{modA}, in the second we inserted $ \Delta_{\eta;\mathcal{B}}^{it}$ for free giving rise to $b$ which is
in $\mathcal{A}$ from the last equation in \eqref{rels}, we used \eqref{embed} in the third line after which we passed $\Delta_{\eta,\psi;\cA}^{-1/2}$ to the right which is allowed since this vector is now in the domain of this operator. We used \eqref{Vpsiprop} in line four and finally $b$ can be rewritten as:
\begin{equation}
\pi^{\mathcal{B}'}(\psi) b =  \Delta_{\psi;\mathcal{B}}^{-it} \pi^{\mathcal{B}}(\psi)   \Delta_{\eta, \psi;\mathcal{B}}^{it}
\end{equation}
using \eqref{coswitch}. This finally leads to \eqref{ghard}. Since both expressions in \eqref{Gt} and \eqref{Gb} involve products of partial isometries we have the following bound on the edges of the strip:
\begin{equation}
|g(it)|, |g(1/2+it)| \leq \| a'  \eta_\cA  \|,
\end{equation}which then extends inside the strip via the Phragmen-Lindel\" of theorem. That theorem also requires the (weaker) bound we derived in \eqref{weakbound}  and it applies inside the closure of the strip. Since $\mathcal{A}' \left| \eta_\cA \right>$ is dense in the Hilbert space we can extend the definition of $g(z)$ to the full Hilbert space, at which point it is a continuous anti-linear functional on all vectors, weakly (hence strongly) holomorphic in $\bS_{1/2}$. This then defines a vector in $\mathscr{H}$ which is then our definition of \eqref{ourvec} on the strip $\bS_{1/2}$. The bound on the norm of this vector follows also from Phragmen-Lindel\" of theorem. 
For the continuity statements we further need the limit of $g(z)$, as $a' \left| \eta \right>$ approaches an arbitrary vector, to be uniform in $z$. This follows easily from the uniform boundedness of $g(z)$ and the Banach-Steinhaus principle.  

(2) The final property \eqref{inducepetz} follows from a short calculation:
\begin{align}
&\left( J_\cA    V_\eta J_{\mathcal{B}} \Delta_{\eta, \psi;\mathcal{B}}^{-it}   \left| \psi_{\mathcal{B}} \right>
,   \Delta_{\eta,\chi;\cA}^{-it}  a_+ \Delta_{\eta,\chi;\cA}^{it}  J_\cA    V_\eta J_{\mathcal{B}} \Delta_{\eta, \psi;\mathcal{B}}^{-it}   \left| \chi_{\mathcal{B}} \right> \right) \nonumber\\
& \quad = \left( J_\cA    V_\eta J_{\mathcal{B}} \Delta_{\eta, \psi;\mathcal{B}}^{-it}   \left| \psi_{\mathcal{B}} \right>
,  \varsigma_\eta^t  ( a_+) \pi^{\cA'}(\psi)  J_\cA    V_\eta J_{\mathcal{B}} \Delta_{\eta, \psi;\mathcal{B}}^{-it}   \left| \chi_{\mathcal{B}} \right> \right)\nonumber \\
& \quad = \left( J_\cA    V_\eta J_{\mathcal{B}} \Delta_{\eta, \psi ;\mathcal{B}}^{-it}   \left| \psi_{\mathcal{B}} \right>
,     {\varsigma_{\eta;\cA}^t  ( a_+)}^{1/2}   \pi^{\cA'}(\psi)   {\varsigma_{\eta;\cA}^t  ( a_+)}^{1/2} J_\cA   V_\eta J_{\mathcal{B}} \Delta_{\eta, \psi;\mathcal{B}}^{-it}   \left| \psi_{\mathcal{B}} \right> \right) \nonumber \\
& \quad \leq \left(     \left| \psi_{\mathcal{B}} \right>
,  \Delta_{\eta, \psi;\mathcal{B}}^{it}   J_{\mathcal{B}}  \left( V_\eta^* J_\cA \varsigma_{\eta, \cA}^t  ( a^+) J_\cA  V_\eta\right) J_{\mathcal{B}} \Delta_{\eta, \psi;\mathcal{B}}^{-it}   \left| \psi_{\mathcal{B}} \right> \right)\nonumber \\
& \quad = \left(     \left| \psi_{\mathcal{B}} \right>
,   \alpha^t_{\eta;\cA}(a_+) \left| \psi_{\mathcal{B}} \right> \right)
= \omega_\psi \circ \iota \circ \alpha_\eta^t(a_+),
\end{align}
where we used \eqref{modflowrel} in the second line, the positivity of $\varsigma_\eta^t  ( a_+)$ in the third line, the bound $\pi^{\cA'}(\psi)  \leq 1$
in the fourth line, the fact that $V_\eta^* \mathcal{A}' V_\eta \subset \mathcal{B}'$ (see \eqref{VAprime}) and again \eqref{modflowrel} for the $\mathcal{B}$ algebra in the fifth line.
\end{proof}

\subsection{Strengthened monotonicity}
\label{subsec:5.5}

\subsubsection{Basic strategy}

We will apply interpolation theory to the vector $\big| \Gamma_\psi(z) \big>$, following the basic strategy of \cite{Junge}. By thm. \ref{thm:gamma} (2) we get the rotated Petz recovered state on the top of the strip at $z= 1/2 + it$, so we need to interpolate to the $L_{1}(\mathcal{A}', \psi)$ norm there where it becomes the fidelity by lem. \ref{lem:fid} (1). Close to $z=0$ we will need to approach the $p=2$ norm (the $\pi(\psi)$ projected Hilbert space norm) by \eqref{eq:h2} where we will show that we can extract the difference in relative entropy. A generalized sum rule, using sub-harmonic analysis, relates the $z=0$ limit to an integral over the fidelities of the $z=1/2+it$ vector. 

Extracting the relative entropy difference is the most difficult part of the proof and requires some modifications to the basic strategy. We proceed by extending the domain of holomorphy to a larger strip so that we can take derivatives at $z=0$ easily. This requires defining a class of states with filtered spectrum for the relative modular operator. We then approach the original state as a limit.  After a continuity argument, we show that this is sufficient to prove a strengthened monotonicity statement for all states with finite $\mathcal{B}$ relative entropy.

\subsubsection{Filtering and continuity}

Our first task will be to extend $\ket{\Gamma_\psi(z)}$ holomorphically into the larger strip $\{-1/2 < {\rm Re} z \leq 1/2\}$. This might not be possible for general $|\psi\rangle$, so to make progress we work with vectors that have approximately bounded spectral support for the relative modular operator $\Delta_{\eta,\psi}$. Thus we now introduce a \emph{filtering} procedure that produces from $|\psi\rangle$ a vector $|\psi_P\rangle$ with approximately bounded spectral support.

For convenience, we work with $|\eta\rangle, |\psi\rangle \in \sH$ in the natural cone, and consider a related vector $|\psi_{P}\rangle$ (which is not  in the natural cone of $\cA$), defined by:
\begin{equation}
\label{psifp}
\left| \psi_{P} \right> = \int_{-\infty}^{\infty}  f_P(t)  \Delta_{\eta,\psi}^{it} \left| \psi \right> \ud t
= \tilde{f}_P(\ln \Delta_{\eta,\psi} ) \left| \psi \right>,
\end{equation}
where $\tilde{f}_P$ is the Fourier transform of a certain function $f_P$ and provides a kind of damping. All modular operators and support projections in this subsection refer to $\cA$, and since we only consider one algebra in this subsection, we drop the subscripts to lighten the notation. Note that $\ln \Delta_{\eta,\psi}$ is defined on $\pi'(\psi)\pi(\psi) \mathscr{H}$ since $\Delta_{\eta,\psi}$ is only invertible there. Away from this subspace the operator acts as $0$. 

We take $f_P$ to have the following properties, motivated by the desire to prove nice continuity statements as $P \rightarrow \infty$. Since we want to think of $P$ as a cutoff, we take $f_P$ to be a scaling function:
\begin{equation}
\label{fscale}
f_P(t) = P f( t P) ,
\end{equation}
and now specify properties of $f(t)$. (Note that the Fourier transform satisfies $\tilde{f}_P(p) = \tilde{f}( p/P)$.)

\begin{definition} \label{defsmooth} We call the function $f$ in \eqref{fscale} a \emph{smooth filtering function} if it satisfies the following properties.
\begin{itemize}
\item[(A)] The Fourier transform of $f$ 
\begin{equation}
\tilde{f}(p) = \int_{-\infty}^{\infty}  e^{ - i t p} f(t) \ud t
\end{equation}
exists as a real and non-negative Schwarz-space function.
This implies that the original function $f$ is Schwarz and has finite $L_1(\mathbb{R})$ norm,  $\| f \|_1 < \infty$.

\item[(B)] $f(t)$ has an analytic continuation to the upper complex half plane such that  the $L_1(\mathbb{R})$ norm 
of the shifted function has $\| f( {\cdot} + i\theta) \|_1 < \infty$
for $0 < \theta < \infty$. 
\end{itemize}
\end{definition}

Note that the Fourier transform of the shifted function satisfies:
\begin{equation}
\widetilde{  f( {\cdot} + i\theta) }(p) = \tilde{f}(p) e^{- \theta p} .
\end{equation}
Examples of such smooth filtering functions include Gaussians as well as the Fourier transform of smooth functions $\tilde{f}$ with compact support. The norms satisfy:
\begin{equation}
\| \tilde{f}_P \|_\infty= \| \tilde{f} \|_\infty \geq \tilde{f}(0) \,, \qquad \|f_P\|_{1} = \|f \|_{1} \geq \|\tilde{f} \|_{\infty} ,
\end{equation}
where the later inequality is well-known as the Hausdorff-Young inequality.

We now establish some properties of the resulting vector $|\psi_P\rangle$:

\begin{lemma}
\label{lem:f}
The filtered vector $\ket{\psi_{P}}$ defined in \eqref{psifp} based on a smooth filtering function $f$, has the following properties:
\begin{enumerate}

\item  $\lim_{P \rightarrow \infty} \left|\psi_{P} \right> = \tilde{f}(0) \left| \psi \right>$ strongly. 
\item  There exists $a_P \in \mathcal{A}$ such that
$
 \left|\psi_{P} \right> = a_{P} \left| \psi \right>,
$
and
 \begin{equation}
 \label{bdafp}
 \| a_{P} \| \leq  \| f \|_{1} ,
 \end{equation}
such that $a_P \pi(\psi) = a_P$.

\item The induced linear functional on $\cA'$ is dominated by
\begin{equation}
\label{linbd}
\omega'_{\psi_{P} }  \leq \|f\|_1^2 \ \omega'_{\psi}, 
\end{equation}
and $\pi'(\psi_P) \leq \pi'(\psi)$. 

\item There exists $a_P' \in \mathcal{A}'$ such that
$
\left| \psi_P \right> = a_P' \left| \eta \right> ,
$
and
$
\| a_P' \|  \leq \| f( {\cdot} + iP/2) \|_1.
$

\item The induced linear functional on the algebra $\mathcal{A}$ is dominated by
\begin{equation}
\label{dom2}
\omega_{\psi_P} \leq \| f( {\cdot} + iP/2) \|_1^2 \ \omega_{\eta}. 
\end{equation}
\end{enumerate}
We need property (A) of def.~\ref{defsmooth} for (1-3) and property (B) for (4-5). 
\end{lemma}
\begin{proof}
(1) Letting $E_{\eta,\psi}(\ud \lambda)$ be the spectral resolution of $\ln \Delta_{\eta,\psi}$, we have
\begin{equation}
\|  \psi_P  - \tilde{f}(0) \psi  \|^2
= \int_{\mathbb R} \left(\tilde{f}(\lambda /P) - \tilde{f}(0) \right)^2 
\left< \psi | E_{\eta,\psi}(\ud \lambda) \psi \right>.
\end{equation}
We can take the pointwise limit $P \rightarrow \infty$ using dominated convergence (since $\tilde{f}$ is a bounded function); this immediately gives the statement.  

(2) We insert $ \Delta_{\psi}^{-it} \pi'(\psi)$ next to $\ket{\psi}$ in the first expression of \eqref{psifp} and find
\begin{equation}
a_{P} = \int_{\mathbb R} f_P(t) \Delta_{\eta}^{it} \Delta_{\psi,\eta}^{-it} \ud t . 
\end{equation}
This is an integral over Connes-cocycles, hence defines an element of $\mathcal{A}$. The operator norm is bounded by 
\begin{equation}
\| a_{P} \| \leq  \int_{\mathbb R} | f_P(t) | \|  \Delta_{\eta}^{it} \Delta_{\psi,\eta}^{-it}  \| \, \ud t
\leq \int_{\mathbb R} | f_P(t) | \, \ud t
= \| f_P \|_1 = \|f \|_1 , 
\end{equation}
since $\Delta_{\eta}^{it} \Delta_{\psi,\eta}^{-it} $ are isometries. 

(3) We establish this via
\begin{equation}
\left< \psi_{P} \right| a_+' \left| \psi_{P} \right>
= \left< \psi \right|  {a_+'}^{1/2} a_{P}^* a_{P} {a_+'}^{1/2} \left|  \psi \right>
\leq \| a_{P} \|^2 \left< \psi \right| a_+' \left| \psi \right> \, , \qquad a_+' \in \mathcal{A}'\,, a_+' \geq 0,
\end{equation}
which gives \eqref{linbd} after using the bound \eqref{bdafp}. The bound on the support projectors follows since $\pi'(\psi_P)$ is the smallest projector $\pi' \in \mathcal{A}'$ that satisfies $\omega_{\psi_P}'( 1 - \pi') = 0$. But $\pi' = \pi'(\psi)$ satisfies this since $0\leq \omega_{\psi_P}'( 1 - \pi'(\psi)) \leq \| f \|_1^2 \omega_\psi'( 1 - \pi'(\psi))  =0$.

(4) Note that $\Delta_{\psi,\eta}^{1/2} \left| \eta \right>=J|\psi\rangle=|\psi\rangle$ since $|\psi\rangle$ is in the natural cone. 
Then, shifting the integration contour as is legal by def.~\ref{defsmooth} (A), 
\begin{align}
\left| \psi_P \right> &= \int_{-\infty}^{\infty}   f_P(t) \Delta_{\eta,\psi}^{it} \Delta_{\psi,\eta}^{1/2} \left| \eta \right>\, \ud t= \int_{-\infty}^{\infty}  f_P(t+i/2) \Delta_{\eta,\psi}^{it} \Delta_{\eta}^{-it} \left| \eta \right>  \, \ud t. 
\end{align}
Note that $\Delta_{\eta,\psi}^{it} \Delta_{\eta}^{-it} $ is a Connes-cocycle for $\cA'$, and hence an element of 
$\cA'$. Now define
\begin{equation}
a_P' =  \int_{-\infty}^{\infty} f_P(t+i/2) \Delta_{\eta,\psi}^{it} \Delta_{\eta}^{-it} \, \ud t \in \mathcal{A}' . 
\end{equation}
Since the Connes-cocycle is isometric, the norm of $a_P'$ may be bounded by
\begin{equation}
\| a_P' \| \leq  \int_{-\infty}^{\infty}  |f_P(t+i/2)| \, \|  \Delta_{\eta,\psi}^{it} \Delta_{\eta}^{-it} \| \, \ud t 
=  \int_{-\infty}^{\infty}  |f(t+ iP/2)| \, \ud t= \| f( {\cdot} + iP/2) \|_1 . 
\end{equation} 

(5) We have
$
\left< \psi_P \right| a_+ \left| \psi_P \right>
= \left< \eta \right| a_+^{1/2} {a_P'}^* a_P' a_+^{1/2} \left| \eta \right>
\leq \|  a_P' \|^2 \left< \eta \right| a_+ \left| \eta \right>, 
$
which gives the statement in view of (4).
\end{proof}

We would now like to see how the relative entropy between $|\eta\rangle, |\psi_P\rangle$ behaves in the limit $P \to \infty$. We will find the conditions on $f$ for which the relative entropy converges to that between $|\eta\rangle, |\psi\rangle$  as $P \rightarrow \infty$. 

\begin{theorem}
\label{lem:finites}
Suppose $|\psi\rangle, |\eta\rangle$ are states on a v. Neumann algebra $\cA$, assumed to be in the natural cone, and suppose $|\psi_P\rangle$ is given by \eqref{psifp} with scaling function \eqref{fscale} satisfying property (A) of def.~\ref{defsmooth} and $\tilde{f}(0)=1$. Then:

\begin{enumerate}
\item
$S(\psi_{P}|\eta) < \infty$. 

\item 
We have
\begin{equation}
\label{lowbd}
- 2 \ln\left(  \|f \|_1  \right) 
+ \limsup_{P \rightarrow \infty} S(\psi_{P}|\eta) 
\leq S(\psi|\eta)  \leq \liminf_{P \rightarrow \infty} S(\psi_{P}|\eta) .
\end{equation}

\item
The relative entropy behaves continuously for $P \rightarrow \infty$,
\begin{equation}
\label{contrel}
\lim_{P \rightarrow \infty} S(\psi_{P}|\eta) = S(\psi|\eta) ,
\end{equation}
iff the Fourier transform of the scaling function, $\tilde{f}(t)$, is a Gaussian centered at the origin. 
\end{enumerate}
\end{theorem}

\begin{remark}
The above statements hold even if $ S(\psi|\eta)  = \infty$ with limits understood as living on the compactified real line.
So for example in this case \eqref{lowbd} or \eqref{contrel} implies that 
$\lim_{P \rightarrow \infty} S(\psi_{P}|\eta)  = \infty$. 
\end{remark}

\begin{proof}
(1) In view of \eqref{linbd}, \cite{Araki2}, thm. 3.6, eq. (3.7), applied to the algebra $\cA'$, gives:
\begin{equation}
\label{ara}
\frac{1}{\Delta'_{\psi_{P},\eta}+ \beta} \geq \frac{1}{\| f \|_1^2 \Delta'_{\psi,\eta}+ \beta} 
\end{equation}
for all $\beta > 0$.\footnote{When applying \cite{Araki2}, thm. 3.6, eq. (3.7) to the commutant $\mathcal{A}'$ using \eqref{modA},
where one switches $\mathcal{A} \leftrightarrow \mathcal{A}'$ as well as any support projectors $\pi \leftrightarrow \pi'$. 
Note further that \cite{Araki2}, thm. 3.6 refers to the natural cone but the specific representative of the linear functional does not affect the modular operators above since $\Delta'_{\xi'_{\psi_{P}},\eta} = \Delta_{\psi_{P},\eta}'$ 
using the notation \eqref{genvec} (now for the commutant).}
The following type of integral representation for the relative entropy is well-known, see e.g. \cite{Petz1993}:
\begin{equation}
\label{intsrel}
S(\psi_{P} | \eta) = \int_0^{\infty}  \left< \psi_{P} \right| \left( - \frac{1}{ \Delta_{\eta,\psi_{P}}^{-1} + \beta} + \frac{1}{\beta+1} \right) \left| \psi_{P} \right> \ud \beta, 
\end{equation}
and integral converges iff the relative entropy is finite.
The bound in \eqref{ara} can be used to bound \eqref{intsrel} from 
above due to the first equation in \eqref{rels} and this gives:
\begin{align}
\nonumber
S(\psi_{P} | \eta) & \leq 2 \ln (\| f \|_1)  \left< \psi_{P} \right| \left. \psi_{P} \right>  + \int_0^{\infty}  \left< \psi_{P} \right| \left( - \frac{1}{ \Delta_{\eta,\psi}^{-1} + \beta} + \frac{1}{\beta+1} \right) \left| \psi_{P} \right> \ud \beta
 \\ 
 & = 2 \ln (\| f \|_1) \left< \psi_{P} \right| \left. \psi_{P} \right>   - \left< \psi \right| \ln \Delta_{\eta,\psi} \left( \tilde{f}_P(\ln \Delta_{\eta,\psi}) \right)^2 \left| \psi \right>    .
 \label{rhsent}
\end{align}
Using the spectral decomposition of $\ln \Delta_{\eta,\psi}$, we can write
\begin{equation}
\label{intrhs}
 \left< \psi \right| \ln \Delta_{\eta,\psi} \left( \tilde{f}_P(\ln \Delta_{\eta,\psi}) \right)^2 \left| \psi \right>   
 = - \int_{- \infty}^{\infty} p \left( \tilde{f}( p/P)\right)^2 \left< \psi \right| E_{\eta, \psi}(\ud p ) \left| \psi \right>.
\end{equation}
This integral converges because $p \tilde{f}( p/P)^2$ is uniformly bounded, by the Schwartz condition in def.~\ref{defsmooth} (A).
Thus the right hand side of \eqref{rhsent} is finite and so we have shown (1). 

(2) Let us continue by first assuming that $S(\psi|\eta) < \infty$. 
Strong convergence of $\psi_P$, lem. \ref{lem:f} (1), guarantees that $\lim_P \left< \psi_{P} \right| \left. \psi_{P} \right>= 1$ since $\tilde{f}(0)= 1$. Now the integral on the right hand
side of \eqref{intrhs} can be split into two parts:
\begin{align}
\lim_{P \rightarrow \infty} \int_{0}^{\pm \infty} |p| \left( \tilde{f}( p/P)\right)^2  \left< \psi \right| E_{\eta,\psi}(\ud p) \left| \psi \right>
=  \int_{0}^{\pm \infty} |p| \left< \psi \right| E_{\eta,\psi}(\ud p) \left| \psi \right>,
\end{align}
where we have applied the dominated convergence theorem to each term using the facts that $\tilde{f}_P(p)$ is bounded  and that the relative entropy is finite.
Taking the lim sup on both sides of \eqref{rhsent} gives the first inequality in \eqref{lowbd}. Lower semi-continuity of relative entropy \cite{Araki2} gives the second inequality. 

If instead $S(\psi|\eta) = \infty$, then we find from lower semi-continuity:
\begin{equation}
\limsup_{P \rightarrow \infty} S(\psi_{P}|\eta) \geq \liminf_{P \rightarrow \infty} S(\psi_{P}|\eta) = \infty, 
\end{equation}
thus the limit must exist on the extended positive real line where it is infinite. This shows (2).

(3) 
Note that  $\|f \|_1 \geq \| \tilde{f}\|_\infty \geq \tilde{f}(0)=1$ so we get the continuity in \eqref{contrel} iff the Hausdorff-Young inequality is saturated and $ \tilde{f}(0)= \| \tilde{f}\|_\infty$.
It was shown by Lieb \cite{Lieb} that the only functions that saturate the Hausdorff-Young bound are in fact Gaussians.  The condition $ \tilde{f}(0)= \| \tilde{f}\|_\infty$ then simply
means the Gaussian $\tilde{f}$ must be centered at the origin. 
\end{proof}

\subsubsection{Updated interpolating vector}

We now consider again our interpolating vector \eqref{ourvec}. With the intention to 
extend the domain of holomorphy, we consider the filtered states $|\psi_P\rangle$ 
instead of $|\psi\rangle$. Although $|\psi_P\rangle$ is not in the natural cone, 
we can still define $\big| \Gamma_{\psi_P}(z) \big>$ in view of rem. \ref{rem:1} (1). This will however by itself not be sufficient: 
It turns out that we also have to apply a projector $\Pi_\Lambda$ to our vectors, so we consider
\begin{equation}\label{eq:upinterp}
 \Pi_\Lambda 
  \big| \Gamma_{\psi_P}(z) \big> \,, \qquad  \Pi_\Lambda  \equiv \int_{-\Lambda}^\Lambda E_{\psi_P}(\ud \lambda) ,
\end{equation}
where $E_{\psi_P}(\ud \lambda)$ is the spectral decomposition of $\ln \Delta_{\psi_P}$, so that $\lim_{\Lambda \to \infty}\Pi_\Lambda = \pi( \psi_P) \pi'(\psi_P)$ in the strong sense. We intend to send the regulators $\Lambda, P \rightarrow \infty$, and in that process we will 
tune $\tilde{f}(0)$ to maintain $\| \left| \psi_{P} \right> \| = 1$, and  require (A) and (B) of def.~\ref{defsmooth}. With those changes in place, we claim the following updated version of thm.~\ref{thm:gamma}. 
\begin{lemma}
\label{lem:newgamma}
\begin{enumerate}
\item
The vector valued function
$
z \mapsto
\Pi_\Lambda  
\big| \Gamma_{\psi_{P}}(z) \big> 
$
can be continued analytically to the extended strip $-1/2 < {\rm Re} z < 1/2$. It is bounded and weakly continuous in the closure. 

\item
Its norm is bounded above by $1$ in the closed upper half strip  $\{0 \leq {\rm Re} z \leq 1/2\}$ and we have the following
estimate in the lower half strip $\{-1/2 \leq {\rm Re} z \leq 0\}$:
\begin{equation}
\label{newbound}
\big\| \Pi_\Lambda  \Gamma_{\psi_{P}}(z)  \big\|  \leq \left( \| f({\cdot} + iP/2) \|_1 e^{\Lambda/2} \right)^{-2{\rm Re} z}.
\end{equation}

\item 
We have
\begin{subequations}
\begin{align}
\label{firstderiv}
&\frac{\ud}{\ud z}  \left( \Pi_\Lambda  \big| \Gamma_{\psi_{P}}(\bar z) \big>, \Pi_\Lambda  \big| \Gamma_{{\psi}_{P}}(z) \big> \right) \bigg|_{z=0} = 2  \frac{\ud}{\ud z} \big< {\psi}_{P} \big| \Gamma_{{\psi}_{P}}(z) \big>\bigg|_{z=0}   \\
& = - 2 \left( S_{\cA}(\psi_{P}| \eta) - S_{\mathcal B}(\psi_{P}| \eta) \right).
\label{secondderiv}
\end{align}
\end{subequations}
\end{enumerate}
\end{lemma}

\begin{proof}
In order for the proof to run  in parallel with that of thm.~\ref{thm:gamma}, we
consider instead of $| \psi_P \rangle$ the corresponding vector $|\xi_{\psi_P}\rangle$ in the natural 
cone of $\cA$. By  rem. \ref{rem:1} (1),  and transformation formulas such as $\Delta_{{\psi}_P}^z = {v'_{\psi_P}}^* \Delta_{\xi_{{\psi}_P}}^z {v'_{\psi_P}}$ (which give corresponding transformation formulas for $\Pi_\Lambda$), we 
find that $\Pi_{\Lambda, \psi_P} \big| \Gamma_{{\psi}_{P}}(z) \big> = {v'_{\psi_P}} \Pi_{\Lambda, \xi_{\psi_P}} \big| \Gamma_{\xi_{\psi_P}}(z) \big>$. The partial isometry ${v'_{\psi_P}}$ is evidently of no consequence for the claims 
made in this lemma. By abuse of notation, we can assume without loss of generality for the rest of this proof that $|\psi_P\rangle$ is in the natural cone. 

(1) Then, as in the proof of thm.~\ref{thm:gamma}, we also use the 
shorthand $\Delta_{\eta_{\mathcal{B}},\psi_{P \,\mathcal{B}}}
=\Delta_{\eta, \psi_P;\mathcal{B}}$ etc. With these notations understood, let us
write out 
\begin{equation}
\label{rewrite}
\Pi_\Lambda \big| \Gamma_{{\psi}_{P}}(z) \big> = 
 \left(\Pi_\Lambda \Delta_{{\psi}_P \, \cA}^{z} \right) \left(  \Delta_{{\psi}_P \, \cA}^{-z} \Delta_{\eta,\psi_P;\cA}^z \right)  V_{\psi_P}\Delta_{\eta, \psi_P;\mathcal{B}}^{-z} \big| \psi_{P \, \mathcal{B}}\big>,
\end{equation}
which is initially defined only for purely imaginary $z$.
We now consider the bracketed operator above: $ \Delta_{{\psi}_P}^{-z} \Delta_{\eta,\psi_P}^z$. It is well known that the majorization condition \eqref{dom2} ensures that this operator
has an analytic continuation to the strip $-1/2 < {\rm Re} z  < 0$. For completeness we give this argument here using a similar approach as in the proof of thm.~\ref{thm:gamma}. 

Thus, we define, dropping temporarily the subscript $\cA$ as all quantities refer to this algebra:
\begin{equation}
G(z) = \left( c^* \big| \psi_P \big> + \left| \zeta \right> , \Delta_{\psi_P}^{-z} \Delta_{\eta,\psi_P}^z d' \big| \eta \big>   \right)
= \left( \Delta_{\psi_P}^{-\bar z} c^* \big| \psi_P \big> ,   \Delta_{\eta,\psi_P}^z d' \left| \eta \right> \right),
\end{equation}
where: $c \in \mathcal{A}, \,\, d' \in \mathcal{A}'$ and $\left| \zeta \right> \in (1-\pi'(\psi_P) ) \mathscr{H}$.
This function is holomorphic in the lower strip $\{-1/2 < {\rm Re} z  < 0\}$ and is continuous in the closure due to Tomita-Takesaki theory.  As in the proof of thm.~\ref{thm:gamma} we can easily derive an upper bound on $|G(z)|$ that is not uniform with $c,d'$.  We can then improve this to a uniform bound using the Phragm\'{e}n-Lindel\"{o}f theorem by checking the top and bottom edges of the strip.
At the top we have:
\begin{align} 
| G(it)  | & \leq  \big\| c^* \big| \psi_P \big> \big\|  \left\| d' \left|  \eta \right> \right\| ,
\end{align}
and at the bottom we need the following calculation:
\begin{align}
G(-1/2+it) &=   \left( \Delta_{\psi_P}^{it} \Delta_{\psi_P}^{1/2} c^* \big| \psi_P \big>,   \Delta_{\eta,\psi_P}^{it} \Delta_{\eta,\psi_P}^{-1/2} d' \left| \eta \right> \right) \nonumber \\
 &=  \left( \Delta_{\psi_P}^{it} \Delta_{\psi_P}^{1/2} c^* \big| \psi_P \big>,   \Delta_{\eta,\psi_P}^{it} J {d'}^* J
 \Delta_{\psi_P}^{-it}  \big| \psi_P \big> \right) \nonumber \\
 &=  \left( \Delta_{\psi_P}^{it} \Delta_{\psi_P}^{1/2} c^* \big| \psi_P \big>,   \Delta_{\eta}^{it} J {d'}^* J
 \Delta_{\psi_P,\eta}^{-it}  \big| \psi_P \big> \right) \nonumber \\
 & = \left( \Delta_{\psi_P}^{it} c^* \big| \psi_P \big>,  J \left( \Delta_{\eta}^{it} J {d'}^* J
 \Delta_{\psi_P,\eta}^{-it} \right)^* J  \big| \psi_P \big> \right) \nonumber
 \\
  & = \left( \Delta_{\psi_P}^{it} c^* \big| \psi_P \big>,  \Delta_{\eta,\psi_P}^{it} d'  \Delta_{\eta}^{it}    \big| \psi_P \big> \right) .  
\end{align}
Consequently,
\begin{align}
| G(-1/2+it)  | & \leq  \big\| \pi(\psi_P ) c^* \big| \psi_P \big> \big\|  \big\|  \pi'(\psi_P ) d'  \Delta_{\eta}^{it}    \big| \psi_P \big> \big\|
\leq  \big\| c^* \big| \psi_P \big> \big\|  \big\|{\varsigma'}_{\eta}^{t}(d')    \big| \psi_P \big> \big\| \nonumber \\
& \leq \| f({\cdot} +iP)  \|_1  \big\| c^* \big| \psi_P \big> \big\|  \left\| {\varsigma'}_{\eta}^{t} (d')  \big| \eta \big> \right\|
= \| f({\cdot} +iP) \|_1  \big\| c^* \big| \psi_P \big> \big\|  \left\| d'    \big| \eta \big> \right\|
\end{align}
where in the first line we dropped the support projectors and defined modular flow on $\cA'$,  ${\varsigma'}_{\eta}^{t}(d')= \Delta_{\eta}^{-it} d' \Delta_{\eta}^{it}$. In the second line we finally used the majorization condition \eqref{dom2} that is true for these filtered states.
These bounds at the edges of the strip, and the weaker bound derived earlier, can be extended into the full strip such that
$
 G(z) \| f({\cdot} +iP)  \|_1^{2z}
$
 is holomorphic and bounded by $1$ everywhere for $ -1/2 \leq {\rm Re}(z) \leq 0$.  Since $c^* \big| \psi_P \big> + \left| \zeta \right>$ and  $d' \big| \eta \big> $
 for all $c \in \mathcal{A}$ and $d' \in \mathcal{A}'$
 are dense, we can extend the definition of the operator $ \Delta_{\psi_P}^{-z} \Delta_{\eta,\psi_P}^z$ to the entire Hilbert space where it remains
 bounded,
 \begin{equation}
 \| \Delta_{\psi_P}^{-z} \Delta_{\eta,\psi_P}^z \| \leq \| f({\cdot} +iP)  \|_1^{-2{\rm Re}z}.
 \end{equation}
 Since the limit on $G(z)$ as $c^* \big| \psi_P \big>$ and  $d' \big| \eta \big> $ approaches two general vectors in the Hilbert space and is uniform in $z$, we get the same
 continuity statement for $ \Delta_{\psi_P}^{-z} \Delta_{\eta,\psi_P}^z$ in the weak operator topology. We also get holomorphy for this operator in the interior of the strip. 
 Note that since $ \Delta_{\psi_P}^{-z} \Delta_{\eta,\psi_P}^z =  (D\psi_P:D\eta)_{-iz} \pi'(\psi_P)$ for the Connes-cocycle $(D\psi_P:D\eta)_{-iz} \in \mathcal{A}$ holds along $z=it$ for real $t$, it continues to take this form in the lower strip.

Now let us turn to the first bracketed operator in \eqref{rewrite}, $\Pi_\Lambda  \Delta_{\psi_P}^{z}$, which is a holomorphic operator (and thus continuous in the strong operator topology) in the entire strip due to the projection on a bounded support of the spectrum of $\ln  \Delta_{{\psi_P}}$. In fact,
the operator norm satisfies $\| \Pi_\Lambda  \Delta_{{\psi_P}}^{z}\| \leq e^{-\Lambda {\rm Re} z}$ for ${\rm Re } z \leq 0$.
Finally we analyze the following vector appearing in \eqref{rewrite}, 
$
 \Delta_{\eta, \psi_P;\mathcal{B}}^{-z} \big| \psi_{P\, \mathcal{B}} \big> 
$
which is holomorphic in $\{-1/2 < {\rm Re} z < 0\}$ and strongly continuous in the closure of this region due to Tomita-Takesaki theory. This vector is also norm bounded by $1$. 

At this stage, we can combine the above holomorphy statements in \eqref{rewrite} 
showing that this vector is analytic in the lower strip $\{ -1/2 < {\rm Re} z < 0\}$. 
For the continuity statement in $z$, note that
an operator that is uniformly bounded and continuous in the weak operator topology such as $ \Delta_{\psi_P,\eta}^{-z} \Delta_{\eta}^z$, 
acting on a strongly continuous vector $ \Delta_{\eta,\psi_{P};\mathcal{B}}^{-z}  \big| \psi_{P \, \mathcal{B}}\big>$ 
gives a weakly continuous vector. 
Similarly, an operator that is continuous in the strong operator topology $ \Pi_\Lambda \Delta_{\psi_P}^{z} $ acting on a weakly continuous vector -- the output of the last statement -- gives a weakly continuous vector. 

Now we use the vector-valued edge of the wedge theorem (see e.g. \cite{Sanders}, app. A), in conjunction
with thm. \ref{thm:gamma}, which already establishes an analytic extension to the upper strip $0< {\rm Re} z < 1/2$.
We thereby extend $\Pi_\Lambda \big| \Gamma_{\psi_{P}}(z) \big>$ holomorphically to the full strip $-1/2 < {\rm Re} z < 1/2$.  

(2) The bound \eqref{newbound} follows by combining the operator norm bounds above. 

(3) Holomorphy at $z=0$ allows us to take the derivative in \eqref{firstderiv} on the bra and ket separately and it is easy to see that they give the same contribution. The equality in \eqref{firstderiv} also relies on $\Pi_\Lambda \big| \psi_{P} \big> =  \big| \psi_{P} \big>$.
The second line \eqref{secondderiv} follows by working with the right hand side of in \eqref{firstderiv} and taking the derivative as a limit along $z = i t$ for $t \rightarrow 0$. This
gives:
\begin{align}
\label{firstlimit}
&\lim_{t \rightarrow 0} \left( \big< {\psi}_{P \, \cA} \big| \Delta_{\eta,{\psi}_{P};\cA}^{it}  V_{\psi_P} 
\Delta_{\eta,\psi_P;\mathcal{B}}^{-it} \big| {\psi}_{P \, \mathcal{B}} \big> -1 \right)/(it) \nonumber \\
=& 
\lim_{t \rightarrow 0} 
\left( \big< {\psi}_{P \, \cA} \big| \Delta_{\eta,{\psi}_{P};\cA}^{it} \big| {\psi}_{P \, \cA} \big> - 1  
\right)/(it) + 
\lim_{t \rightarrow 0}
 \left( 
\big<  {\psi}_{P \, \mathcal{B} } \big| \Delta_{\eta, \psi_P;\mathcal{B}}^{-it} \big| \psi_{P \, \mathcal{B}} \big> - 1  
\right)/(it) \nonumber \\
=&- S_{\cA}(\psi_{P}| \eta) + S_{\mathcal B}(\psi_{P}| \eta),
\end{align}
where the later limits can be shown to exist when the $\psi_P$ relative entropies are finite, as is indeed the case by thm. \ref{lem:finites} (1), see \cite{Petz1993}, thm. 5.7. The first equality in \eqref{firstlimit} can be shown more explicitly by subtracting the two sides and observing that this is an inner product on two vectors. After applying the Cauchy-Schwarz inequality, one again uses the finiteness of $\psi_P$ relative entropy, by thm. \ref{lem:finites} (1), to show that this difference vanishes in the limit:
\begin{align}
& \lim_{t \rightarrow 0} \frac{ \left| \left(\Delta_{\eta,{\psi}_{P};\cA}^{-it}   \big| {\psi}_{P \, \cA} \big> -   
\big| {\psi}_{P \, \cA} \big> , \,\,  V_{\psi_P} \Delta_{\eta, \psi_P;\mathcal{B}}^{-it} \big| 
{\psi}_{P \, \mathcal{B}} \big> -  \big| {\psi}_{P \, \cA} \big> \right) \right|^2}{t^2} \nonumber \\
\le & \lim_{t \rightarrow 0}  \frac{2 {\rm Re} \left( 1 -  \big< {\psi}_{P \, \cA} \big| \Delta_{\eta,{\psi}_{P}; \cA}^{it} \big| {\psi}_{P \, \cA} \big>\right)}{t} \frac{2 {\rm Re} \left(1 - \big< ( {\psi}_{P \mathcal{B} } \big| \Delta_{\eta, \psi_P;\mathcal{B}}^{-it} \big|  {\psi}_{P \, \mathcal{B} } \big>\right)}{t}   = 0. 
\end{align}
\end{proof}

\subsubsection{$L_p$ norms of updated interpolating vector}

We now study $L_p$ norms of the updated interpolating vector \eqref{eq:upinterp} and its 
limits as $P,\Lambda \to \infty$, $z\to 0$ and $p \to 1$ or $p \to 2$. First we consider $p=1$.

\begin{lemma}
\label{lem:p1} 
\begin{enumerate}
\item
The $L_1(\cA',\psi_P)$-norm of   \eqref{eq:upinterp} for $z = 1/2 + i t$ satifsfies:
\begin{equation}
\lim_{\Lambda \rightarrow \infty} \left\| \Pi_\Lambda  \Gamma_{{\psi}_{P}}(1/2+ i t) \right\|_{1,{\psi}_{P}}^{\mathcal{A}'} =
\left\|  \Gamma_{{\psi}_{P}}(1/2+ i t) \right\|_{1,{\psi}_{P}}^{\mathcal{A}'} 
\leq  F(\omega_{\psi_{P}}, \omega_{\psi_{P}} \circ \iota \circ \alpha_{\eta}^t ) 
\end{equation}
where $ \alpha_{\eta}^t $ is the rotated Petz map defined in \eqref{petzfaith}. 

\item
We have
\begin{equation}
\label{limP}
\lim_{P \rightarrow \infty}  F(\omega_{\psi_{P}}, \omega_{\psi_{P}} \circ \iota \circ \alpha_{\eta}^t ) =
 F(\omega_{\psi}, \omega_{\psi} \circ \iota \circ\alpha_{\eta}^t ) 
\end{equation}
\end{enumerate}
\end{lemma}
\begin{proof}
(1) For the first equality, we need an appropriate continuity property of the $L_1$-norm which is provided 
in lem. \ref{lem:Fcont},  app.~\ref{sec:lem:p1}. It shows that strong convergence of the 
vectors implies the convergence of the $L_1$ norm. 
For the limit $\Lambda \rightarrow \infty$, this follows from the strong convergence of $\Pi_\Lambda$ to $\pi'(\psi_P)\pi(\psi_P)$. In fact, we can drop these support projectors because by definition $\pi'({\psi_P}) \big| \Gamma_{{\psi}_{P}} (z)\big> = \big| \Gamma_{{\psi}_{P}} (z)\big>$ and also because the $L_p$ norms satisfy \eqref{projequal}. 

Next, lem.~\ref{lem:fid} (1) gives $
\left\|  \Gamma_{{\psi}_{P}}(1/2+ i t) \right\|_{1,{\psi}_{P}}^{\mathcal{A}'} 
= F(\omega_{\psi_P},\omega_{\Gamma_t} ),
$
where we use the shorthand  $|\Gamma_t \rangle \equiv |\Gamma_{\psi_P}(1/2+it)\rangle$. 
Now we use the majorization condition on $\omega_{\Gamma_t}$ \eqref{inducepetz}, 
in conjunction with the concavity of the fidelity \cite{UhlmannFidelity}:
\begin{align}
\label{convexFid}
F(\omega_{\psi_P}, \omega_{\psi_P} \circ \iota \circ \alpha_\eta^t) 
&= F(\omega_{\psi_P},\omega_{\Gamma_t} + (\omega_{\psi_P} \circ \iota \circ \alpha_\eta^t- \omega_{\Gamma_t} )) \nonumber \\ 
& \geq F(\omega_{\psi_P}, \omega_{\Gamma_t} ) +  F(\omega_{\psi_P},  (\omega_{\psi_P} \circ \iota \circ \alpha_\eta^t- \omega_{\Gamma_t} )) \nonumber
\\ & \geq F(\omega_{\psi_P}, \omega_{\Gamma_t} )
\end{align}
This completes 
the proof of (1).
 
(2) We use the fact that, where the fidelity $F(\omega_{\psi_P}, \omega_{\psi_P} \circ \iota \circ \alpha_\eta^t)$ is concerned, we can pick another vector that gives the same linear functional. We can replace:
\begin{equation}
F(\omega_{\psi_P}, \omega_{\psi_P} \circ \iota \circ \alpha_\eta^t ) = \big\|    \Delta_{\eta;\cA}^{it}  J_\cA    V_\eta J_{\mathcal{B}} \Delta_{\eta;\mathcal{B}}^{-it}   {\psi}_{P  \mathcal{B}} 
\big\|_{1,\psi_P}^{\mathcal{A}'}.
\end{equation}
Then, in view of lem. \ref{lem:Fcont},  app.~\ref{sec:lem:p1}, 
 we only need establish the strong convergence of $\big| {\psi}_{P \, \mathcal{B}} \big>$ and of $\big|\psi_{P \, \cA} \big>$
 as $P \to \infty$, 
and this follows by combining lem.~\ref{lem:f} (1) and eq. \eqref{eq:PS} 
[remembering the notations \eqref{eq:notation}]. 
\end{proof}

Next, we consider simultaneously approaching $p=2$ and $z= 0$.

\begin{lemma}
\label{lem:p2} 
We have
\begin{align}
\lim_{ \theta \rightarrow 0} \frac{1}{\theta} \ln \left\| \Pi_{\Lambda} \Gamma_{{\psi}_{P}}(\theta) \right\|_{p_\theta, {\psi}_{P}}^{\mathcal{A}'}
&= \lim_{\theta \rightarrow 0} \frac{1}{2\theta} \ln \left( \Pi_\Lambda \left| \Gamma_{{\psi}_{P}}(\theta) \right>, \Pi_\Lambda \left| \Gamma_{{\psi}_{P}}(\theta) \right> \right) \\
&= -  \left( S_{\cA}(\psi_{P}| \eta) - S_{\mathcal B}(\psi_{P}| \eta) \right) \nonumber
\label{limp2}
\end{align}
with $p_\theta = 2/(1+2 \theta)$. 
\end{lemma}

\begin{proof}
Define the normalized vector
\begin{equation}
\left| \zeta_\theta \right> \equiv \frac{\Pi_{\Lambda}\big| \Gamma_{\psi_{P}}(\theta) \big> }{\| \Pi_{\Lambda}\Gamma_{\psi_{P}}(\theta)  \|}.
\end{equation}
We can then use lem.~\ref{lem:newgamma}, \eqref{firstderiv} to show that:
\begin{equation}
\lim_{\theta \rightarrow 0^+} \frac{ \| \zeta_\theta - \psi_P \|^2}{\theta} =0.
\end{equation}
So we can apply the ``first law'' \eqref{firstsand} for the $L_p$ norms in lem.~\ref{firstlaw} to $|\zeta_\theta\rangle$, to conclude 
\begin{equation}
\lim_{\theta \rightarrow 0^+}  \frac{1}{\theta} \ln \| \zeta_\theta\|_{p_\theta,\psi_P}^{\mathcal{A}'} =0,
\end{equation}
since $p_\theta = 2/(1+2\theta)$ satisfies the assumptions of lem.~\ref{firstlaw}. 
The $L_p$ norms are homogenous so we can pull out the normalization:
\begin{equation}
\lim_{\theta \rightarrow 0^+}  \frac{1}{\theta} \ln \| \Pi_{\Lambda} \Gamma_{\psi_{P}}(\theta)  \|_{p_\theta,\psi_P}^{\mathcal{A}'}
= \lim_{\theta \rightarrow 0^+}  \frac{1}{\theta} \ln \| \Pi_{\Lambda}\Gamma_{\psi_{P}}(\theta) \|,
\end{equation}
and this gives the desired answer after applying \eqref{firstderiv} again.
\end{proof}
The last ingredient that we will need is an interpolation theorem for the Araki-Masuda $L_p$
norms on a v. Neumann algebra:

\begin{lemma}
\label{lem:hirsch}
Let $|G(z)\rangle$ be a $\sH$-valued holomorphic function on the strip $\bS_{1/2}=\{0<{\rm Re}z<1/2\}$ 
that is uniformly bounded in the closure, $|\psi\rangle \in \sH$ a possibly non-faithful state of a
sigma-finite v. Neumann algebra $\mathcal M$ in standard form acting on $\sH$. Then, for $0<\theta<1/2$, 
\ben
\frac{1}{p_\theta} = \frac{1-2\theta}{p_0} + \frac{2\theta}{p_1}
\een
with $p_0,p_1 \in [1,2]$, we have 
\begin{align}
\label{himp}
& \ln \left\| G(\theta)\right\|^{\mathcal M}_{p_\theta, \psi} \\
  \leq  & \int_{-\infty}^{\infty} \ud t 
  \left(
 (1-2 \theta)  \alpha_\theta(t) \ln \left\| G(it) \right\|^{\mathcal M}_{p_0, \psi} + (2\theta)  \beta_\theta(t) 
  \ln  \left\| G(1/2+it) \right\|^{\mathcal M}_{p_1, \psi} \right),
  \nonumber
\end{align}
where
\begin{equation}
\alpha_\theta(t) = \frac{ \sin(2\pi\theta)}{(1-2\theta)(\cosh(2\pi t ) - \cos(2\pi \theta)) }\,,
\qquad \beta_\theta(t) = \frac{ \sin(2\pi\theta)}{2 \theta(\cosh(2\pi t ) + \cos(2\pi \theta)) }.
\end{equation}
\end{lemma}

\begin{proof}
See app.~\ref{sec:lem:hirsch}. In the commutative setting this is closely related to the Stein interpolation theorem \cite{Stein}.
In the non-commutative setting, a proof appears for type I factors and the usual 
non-commutative Schatten $L_p$ norms in \cite{Junge}.  We will make sure that it 
works in the setting of the Araki-Masuda $L_p$ norms defined in \eqref{AMdef} with reference to a possibly non-faithful state. 
\end{proof}

\subsection{Proof of Theorems \ref{thm1} and \ref{thm2}}
\label{subsec:5.6}

We close out this long section by combining the above auxiliary results into proofs of the main theorems.

\begin{proof}[Proof of Theorem ~\ref{thm1}]
Given the two normal states $\rho,\sigma$ we consider as above representers $|\psi\rangle, |\eta\rangle$ in the natural cone. From this we construct the filtered vector $\ket{\psi_P}$ as in \eqref{psifp}.
We then apply lem.~\ref{lem:hirsch} with $p_1=2, p_0=1$, ${\mathcal M}=\cA'$, 
$|G(z)\rangle= \Pi_\Lambda \big| \Gamma_{{\psi}_{P}}(z) \big>$
and use that the $L_2$ norm is actually the (projected) Hilbert space norm, see eq. \eqref{eq:h2}, so
\begin{equation}
\left\| \Pi_\Lambda \Gamma_{{\psi}_{P}}(it)  \right\|^{\mathcal{A}'}_{2, \psi_P} = \left\| \Pi_\Lambda \Gamma_{{\psi}_{P}}(it)  \right\|  \leq 1.
\end{equation}
Taking the limit $\theta \rightarrow 0^+$ with the aid of lem.~\ref{lem:p2} we have:
\begin{align}
  S_{\cA}(\psi_{P}| \eta) - S_{\mathcal B}(\psi_{P}| \eta)  &\geq
 - 2\lim_{\Lambda \to \infty}
 \int_{-\infty}^{\infty}  \beta_0(t) \ln \left\| \Pi_\Lambda  \Gamma_{{\psi}_{P}}(1/2+it)  \right\|_{1, {\psi}_{P}}^{\cA'} \ud t
 \nonumber \\
 & =  - 2\int_{-\infty}^\infty \beta_0(t) \ln \left\| \Gamma_{{\psi}_{P}}(1/2+it)  \right\|_{1, {\psi}_{P}}^{\cA'}  \ud t 
 \nonumber \\
  & \geq  - 2 \int_{-\infty}^\infty \beta_0(t) \ln F\left(\omega_{\psi_{P}},\omega_{\psi_{P}} \circ \iota \circ \alpha_\eta^t \right)\ud t, 
\end{align}
where the limit exits due to lem.~\ref{lem:p1} (1) and where we have used the monotonicity of $\ln$.
Taking the limit $P \rightarrow \infty$ we get in view of lem.~\ref{lem:p1} (2), thm.~\ref{lem:finites} (3) 
for a Gaussian filtering function  satisfying (A) and (B) of def. \ref{defsmooth} and lower semi-continuity of the 
$\mathcal B$ relative entropy that
\begin{equation}
 S_{\cA}(\psi| \eta) - S_{\mathcal B}(\psi| \eta) 
 \geq - 2 \int_{-\infty}^\infty \beta_0(t) \ln F\left(\omega_{\psi},\omega_{\psi} \circ \iota \circ \alpha_\eta^t \right)  \ud t. 
\end{equation}
We can then re-write the answer in terms of the original states $\rho,\sigma$ and we arrive at \eqref{streng}. (Recall that we are using $\alpha_\eta^t = \alpha_\sigma^t$ interchangeably.) 
\end{proof}

Thm.~\ref{thm1} forms the basis of the next proof:

\begin{proof}[Proof of Theorem \ref{thm2}]
Since all states $ \rho_i \in \sS$ have finite relative entropy with respect to $\sigma \in \sS$ we learn that $\pi(\rho_i) \leq \pi(\sigma)$.
This implies, via lem.~\ref{lem:proj}, (in particular \eqref{equalfunc}) that if $\iota_\pi(\mathcal{B}_\pi) \subset \mathcal{A}_\pi$ is $\epsilon$-approximately sufficient for $\sS_\pi$ then 
$\iota(\mathcal{B}) \subset \mathcal{A}$ is $\epsilon$-approximately sufficient for $\sS$. Here
\begin{equation}
\sS_\pi = \{ \rho \circ \Phi   \in (\mathcal{A}_\pi)_\star :  \rho \in \sS \},
\end{equation}
and we have used \eqref{defPhi}. The recovery channel $\alpha_{\sS}$ is derived from the recovery channel for $\iota_\pi(\mathcal{B}_\pi) \subset \mathcal{A}_\pi$.
This later recovery channel $\alpha_{\sS_\pi}$ then pertains to the ``faithful'' version of this theorem, and is derived from Theorem~\ref{thm1}, as we will show below.
In this way we can proceed by simply assuming that $\sigma$ is faithful for $\mathcal{A}$, now without loss of generality. 
In particular we may take \eqref{recover} to be determined by the faithful Petz map in \eqref{petzfaith}.

In the faithful case we first check that the map \eqref{recover} is indeed a recovery channel. 
This follows since $\alpha_\sigma^t$ are recovery channels for each $t \in \mathbb{R}$  (
generalizing the results in \cite{Petz1} to non-zero $t$) and so the weighted $t$ integral is also clearly unital and completely positive.  

We now check the continuity property of \eqref{recover}. The integral is rigorously defined as follows. For all $a\in\mathcal{A}$ the function
$
t \mapsto \alpha_\sigma^t(a)
$
is continuous in $t$ in the ultra-weak topology (thus Lebesgue measurable) and bounded on $\mathbb{R}$. 
So 
\begin{equation}
\mathcal{B}_\star \owns \rho \mapsto \int_{\mathbb{R}} p(t) \rho( \alpha_\sigma^t(a)) \ud t \in \mathbb{C}
\end{equation}
gives a continuous linear functional and thus
defines an element in $\mathcal{B}$ (the continuous dual of the predual) that we call $\alpha_{\sS}(a)$. Continuity in the linear functional norm follows from the convergence of
the following integral:
\begin{equation}
\int_{\mathbb{R}} p(t) \| \alpha_\sigma^t(a) \|  \, \ud t \le \| a \| .
\end{equation}
This also guarantees that the resulting operator $\alpha_\sS(a)=\int_{\mathbb{R}} p(t) \alpha_\sigma^t(a) \ud t$ is a bounded operator:
\begin{equation}
\| \alpha_{\sS}(a) \| = \sup_{\rho \in \mathcal{A}_\star} \frac{| \int_{\mathbb{R}} p(t) \rho( \alpha_\sigma^t(a)) \ud t |}{\| \rho \|}
\leq \int_{\mathbb{R}} p(t) \| \alpha_\sigma^t(a) \| \ud t \le \|a\|. 
\end{equation}
We need to check the ultraweak continuity of $a \mapsto \alpha_{\sS}(a)$. 
For all $\rho \in \mathcal{B}_\star$ we define the integral
\begin{equation}
\label{intpre}
\int_{\mathbb{R}} p(t) \rho \circ \alpha_\sigma^t \, \ud t
\end{equation}
in much the same way as above, as a Lebesgue integral on continuous functions valued in $\mathcal{A}_\star$. That is, 
the evaluation of this expression on $a \in \cA$
defines an ultraweakly continuous functional on $\mathcal{A}$. This follows since the sequence
\begin{equation}
\label{pint}
\int_{\mathbb{R}} p(t)  \rho \circ \alpha_\sigma^t(a_n)  \ud t
\end{equation}
converges to the integral of the pointwise limit by the dominated convergence theorem, as
$
p(t)  |\rho \circ \alpha_\sigma^t(a) | \leq p(t) \| \rho \| \|a  \|
$
is integrable. Putting all the pieces together we find that
\begin{equation}
a \mapsto \alpha_\sS(a)= \int_{-\infty}^{\infty} p(t) \alpha_\eta^t(a) \, \ud t
\end{equation}
is ultraweakly continuous, since for all $\rho \in \mathcal{B}_\star$,
\begin{equation}
\rho\left( \int_{\mathbb{R}} p(t) \alpha_\sigma^t(a_n-a) \right) \ud t
\equiv  \int_{\mathbb{R}} p(t) \rho( \alpha_\sigma^t(a_n-a)) \ud t  =  \int_{\mathbb{R}} 
p(t)  \rho \circ \alpha_\sigma^t(a_n-a) \ud t 
\end{equation}
converges to zero whenever $a_n \rightarrow a$ ultraweakly. 

The proof is then completed by rewriting thm.~\ref{thm1} using the concavity of fidelity. 
For this, we require a version of Jensen's inequality for the convex functional 
$\sigma \mapsto F(\rho, \sigma) $ on normal states on $\cA$ with respect to the 
measure $p(t)\ud t$. This would give us
\begin{equation}
\label{contcav}
\int_\mathbb{R}  F(\rho, \rho\circ \iota \circ \alpha_\sigma^t) p(t) \ud t 
\leq F\left(\rho, \int_\mathbb{R} \rho\circ \iota \circ \alpha_\sigma^t  \, p(t) \ud t \right)
\end{equation}
where $\rho $ is a state in $\mathcal{A}_\star$. Then thm.~\ref{thm1} becomes:
\begin{equation}
- 2 \ln F( \rho, \rho \circ \iota \circ \alpha_{\sS}) \leq S_{\cA}(\rho|\sigma) - S_{\cB}(\rho|\sigma), 
\end{equation}
which implies that $\mathcal{B}$ is $\epsilon$-approximately sufficient as claimed by the theorem.   

We are not aware of a proof for Jensen's inequality for convex functionals
of a Banach space valued random variable that would apply straight away to the 
case considered here. In particular, it is not evident that the integrals in question 
can be approximated by Riemann sums in the general case, as was done in \cite{Junge}. 
So we now demonstrate \eqref{contcav} by 
a more explicit argument using the detailed structure of the fidelity.

Consider the Hilbert space $\sY=L_2(\mathbb{R};\mathscr{H};p(t) \ud t) \cong \sH
\bar \otimes  L_2(\mathbb{R};p(t) \ud t)$ of strongly measurable square integrable functions valued in $\mathscr{H}$. Vectors $|\Upsilon\rangle$ 
in this space are (equivalence classes of) functions $t \mapsto |\Upsilon_t\rangle$. $\sY$ is evidently 
a module for $\cA$. We denote this v. Neumann algebra 
by $\cA \otimes 1$ since it acts trivially in the second $L_2$ tensor factor of $\sY$. 
Now define the fidelity as:
\begin{equation}
\label{eq:var}
F_{\mathcal{A} \otimes 1 }(\Psi,\Upsilon)
= \sup_{Y' \in (\mathcal{A} \otimes 1)'\,, \,\, \| Y' \| \leq 1 }  \left| \left< \Psi \right| Y' \left| \Upsilon \right> \right| .
\end{equation}
We next formulate a lemma that will allow us to complete the proof. 
\begin{lemma}
Let $\left| \Upsilon \right>, \left| \Psi \right> \in \sY$ induce linear functionals  
on $\mathcal{A} \otimes 1$ such that
\begin{equation}
\left< \Upsilon \right| a_+ \otimes 1 \left| \Upsilon \right> \leq \sigma(a_+) ,
\qquad \left< \Psi \right| a_+ \otimes 1 \left| \Psi \right> \leq \rho (a_+). 
\end{equation}
where $a_+$ is an arbitrary non-negative element in $\mathcal{A}$ and $\sigma,\rho$
states on $\cA$. Then if  
$\left| \Upsilon_t \right>, \left| \Psi_t \right>$ are strongly continuous then  
$F(\Upsilon_t, \Psi_t)$ is continuous, and we have
\ben
\label{eq:Flem}
F(\sigma, \rho) \ge \int_{\mathbb R} F(\Upsilon_t, \Psi_t) p(t) \ud t.
\een
\end{lemma}
\begin{proof}
 If  $\left| \Upsilon_t \right>, \left| \Psi_t \right>$ are strongly continuous then  
$F(\Upsilon_t, \Psi_t)$ is continuous in $t$ by \eqref{contboth}, and since the fidelity is 
the $L^1$ norm, see app. C. 

The idea is now to construct 
a suitable family of elements $y_t' \in \cA'$. This family should be chosen at the same time so as to satisfy:
(i) $\| y_t' \| \leq 1$, (ii) $t \mapsto y_t'$ is strongly continuous, (iii) in the sup definition 
of the fidelity, \eqref{uhl2} we are suitably close to saturating the supremum in the sense that $F(\Upsilon_t,\Psi_t)$
is approximately $|\langle \Upsilon_t | y'_t \Psi_t \rangle|$.
Then (ii) implies that $y'_t|\Psi_t\rangle$ is weakly measurable and thus strongly measurable by the Pettis measurability theorem, see e.g. \cite{Pettis}, thm. 3.1.1.\footnote{This theorem applies even without assuming $\sH$ to be separable since the image $\{y'_t|\Psi_t\rangle : t \in {\mathbb R}\}$ is a separable open subset of $\mathscr{H}$, in the norm topology, by strong continuity.} By (i) 
we then see that the map $y'_t|\Psi_t\rangle$ is in the Hilbert space $\sY$ because  boundedness $y_t'$ clearly 
implies that it is square integrable. 
(ii) holds for instance if the function $y'_t$ is continuous in the norm topology, and 
we will attempt to choose it in this way. Then $y'_t$, as a function, will define an element $Y$ in $(\mathcal{A} \otimes 1)'$
that can be used in the variational principle \eqref{eq:var}. 
We must therefore have, using concavity of the fidelity in the same manner as in  \eqref{convexFid},
\begin{equation}
\label{Fff}
F(\sigma,\rho ) \geq F_{\mathcal{A} \otimes 1 }(\Psi,\Upsilon) \geq \left| \int_{\mathbb R}  
\langle \Upsilon_t | y_t' | \Psi_t \rangle p(t) \ud t \right|, 
\end{equation}
using the variational principle \eqref{eq:var} to obtain the last inequality, and using that the fidelity only 
depends on functionals in the first. The evident strategy is now to make our choice (iii) of of the function $y_t'$
in such a way that the right side is close to the right side of \eqref{eq:Flem}, while being continuous in the operator norm topology and while satisfying $\|y'_t\|<1$, so that (i) and (ii) hold as discussed.

To this end, consider the open unit ball in $\cA'$ in the norm topology,
$
\mathcal{A}'_1 \equiv \{ x' \in  \mathcal{A}' : \| x' \| < 1 \}.
$
For all $t$ we define next a subset $\mathcal{X}_t' \subset \mathcal{A}'_1$ by
\begin{equation}
\label{sand}
\mathcal{X}'_t \equiv \mathcal{A}'_1 \cap \{ x' \in \mathcal{A}' : \left|\left< \Psi_t \right| x' \left| \Upsilon_t \right>- 
F(\Psi_t, \Upsilon_t) \right|  < \epsilon \}. 
\end{equation}
This set is open in the norm topology because the second set on the right hand side of \eqref{sand} is open in the weak operator topology and so it is open in the norm topology, too.  
It is non empty since we know that in the sup definition of fidelity it is sufficient to take $\| x' \| < 1$ and still achieve $F(\Psi_t, \Upsilon_t)$. 

We will be interested in the norm closures $\overline{\mathcal{X}_t'}$.
What we then need to do is select a function from this set $y_t' \in \overline{\mathcal{X}_t'}$ that varies continuously in 
the operator norm. 
This problem can be solved by the Michael selection theorem \cite{Michael}. Indeed, we can consider the mapping
$
t \in \mathbb{R} \rightarrow \overline{\mathcal{X}_t'} \in 2^{\mathcal{A}'}
$
as a map from the paracompact space $\mathbb{R}$ to subsets of $\mathcal{A}'$ thought of as a the Banach space (with the operator norm).
If it can be shown that the sets $\overline{\mathcal{X}_t'}$ are nonempty closed and convex and that this map is ``lower hemicontinuous'', then by the Micheal selection theorem, there is a continuous selection $y_t' \in \mathcal{X}_t'$ as we require.  

We have seen that the sets are closed and nonempty. Convexity follows from 
\begin{align}
\left| \left< \Psi \right| p_1 x'_1 + p_2 x'_2  \left| \Upsilon \right> - F (p_1+p_2) \right| &\leq p_1 \left| \left< \Psi \right|  x'_1 \left| \Upsilon \right> - F \right|
+ p_2 \left| \left< \Psi \right|  x'_2  \left| \Upsilon \right> - F \right| \nonumber \\
\| p_1 x'_1 + p_2 x'_2 \| &\leq p_1 \| x_1' \| + p_2 \| x_2' \|
\end{align}
where the first equation is schematic but is hopefully clear, and where $p_1,p_2 \ge 0, p_1+p_2=1$. 
This implies that $\mathcal{X}_t'$ is convex and hence its closure is also convex. 

Lower hemicontinuity at some point $t$ is the property that for any open set $\cV \subset \mathcal{A}'$ that intersects $\overline{\mathcal{X}'_t}$ there exists a $\delta$ such that
$\overline{\mathcal{X}'_{t'}} \cap \cV \neq \emptyset$ for all $|t-t'| < \delta$. We see this for the case at hand as follows. Take $\cV$ satisfying the assumption, and note that $\cV \cap \mathcal{X}'_{t}$ is also non empty. Pick a $y' \in \cV \cap \mathcal{X}'_{t}$. There exists an $\epsilon' < \epsilon$ such that:
\begin{equation}
\left|\left< \Psi_t \right| y' \left| \Upsilon_t \right>- F(\Psi_t,\Upsilon_t) \right| < \epsilon' < \epsilon.
\end{equation}
Then, by the strong continuity of $\left| \Upsilon_t \right>$ resp. $|\Psi_t\rangle$ 
and continuity of $F(\Psi_t,\Upsilon_t)$,
we see that this condition is stable: Given $\epsilon - \epsilon' >0$ there does indeed exist a $\delta$ such that 
\begin{equation}
\left|\left< \Psi_{t'} \right| y' \left| \Upsilon_{t'} \right>- F(\Psi_{t'},\Upsilon_{t'}) \right| <  \epsilon \,, \qquad \forall |t- t'| < \delta
\end{equation}
which implies that $y' \in  \cV \cap \mathcal{X}'_{t'} \subset  \cV \cap \overline{\mathcal{X}'_{t'}}$ as required. 

From Michael's theorem we therefore get the desired  norm continuous $y_t'$ satisfying
\begin{equation}
\left|\langle \Psi_t | y_t' |  \Upsilon_t \rangle- F(\Psi_t,\Upsilon_t) \right|  \leq \epsilon 
\end{equation}
for all $t$. Using that the fidelity is real and \eqref{Fff} and that $\epsilon$ can be made arbitrarily small then readily implies the lemma.  
\end{proof}

We now use this lemma with $|\Upsilon_t\rangle := |\Gamma_\psi(i/2+t)\rangle$, which is weakly continuous 
by thm. \ref{thm:gamma} (1). Actually, it is even strongly continuous since it is given by the product 
of bounded operators and $\Delta_{\eta; \cA}^{it}, \Delta_{\eta; \cB}^{it}$, which are strongly continuous as 
they are 1-parameter groups of unitaries generated by a self-adjoint operator by Stone's theorem, 
see e.g. \cite{specth}, sec. 5.3. 
We also take $|\Psi_t\rangle = |\psi\rangle$, which is obviously strongly continuous as it is just constant. 
Then $|\Upsilon\rangle$ induces a state dominated by 
$\rho \circ \iota \circ \alpha_{\sS}$, by thm. \ref{thm:gamma} (2), and $|\Psi\rangle$ induces $\rho$
by definition, and $|\Upsilon_t\rangle$ induces 
$\rho \circ \iota \circ \alpha_{\sigma}^t$. We thereby arrive at the concavity result \eqref{contcav}, and this concludes the proof of thm. \ref{thm2}. 
\end{proof}

\section{Examples}

Here we illustrate our method and results in two representative examples.

\subsection{Example: finite type-I algebras}\label{finitealgebra}

To compare our method to that of \cite{Junge} in the subalgebra case, we work out our interpolating vector 
\eqref{ourvec} in the matrix algebra case. Thus let $\mathcal{A} = M_{n}({\mathbb C})$ and
$\mathcal{B} = M_{m}({\mathbb C})$, ${\mathcal C}=\mB' \cap \cA$, embedded as the subalgebra $b \mapsto \iota(b) = b \otimes 1_{\mathcal C}$ where $n = m \times k$ and these integers  label the size of the matrices. We will work in the standard Hilbert space ($\sH \simeq M_{n}({\mathbb C}) \simeq
{\mathbb C}^{n *} \otimes  \mathbb{C}^n$) and identify state functionals such as $\sigma$ with density matrices. 
So for example $\sigma_{\cA} \in M_{n}({\mathbb C})$, and we assume for simplicity that this has full rank (faithful state).

$\sH \simeq M_{n}({\mathbb C})$ is both a left and right module for $\cA$, 
\begin{equation}
l(m_1) \left| m_2 \right> = \left| m_1 m_2 \right>\, \qquad r(m_1) \left| m_2 \right> = \left| m_2 m_1 \right>, 
\end{equation}
and the inner product on $\sH$ is the Hilbert-Schmidt inner product. The natural cone of $\cA$ is 
defined to be the subset of positive semi-definite matrices in $\sH$.
The modular conjugation and relative modular operators (of $\cA$)
associated with this natural cone are:
\begin{equation}
 J \big| m \big> = \big| m^* \big> \, \qquad \Delta_{\eta,\psi} = l(\sigma_\cA) r(\rho_\cA^{-1}),
\end{equation}
where we invert the density matrix $\rho_\cA$ on its support.  The natural cone vectors correspond to the unique positive square root of the corresponding density matrix,
now thought of as pure states in the standard Hilbert space. So  $|\psi_\cA \rangle = \big| \rho_\cA^{1/2} \big>$ and $|\psi_{\mathcal{B}} \rangle = \big| \rho_\mB^{1/2} \big>$. The embedding is:
\begin{equation}
V_\eta = r(\sigma_\cA^{1/2})T^* r(\sigma_\mB^{-1/2}) \,, \qquad T^* (m_\mB) = m_\mB \otimes 1_{\mathcal C}
\end{equation}
Using these replacements it is easy to compute our interpolating vector \eqref{ourvec} $|\Gamma_\psi(z)\rangle$ by starting with the expression in \eqref{Gt}
\begin{align}
\left| \Gamma_\psi(z) \right> = \left| \sigma_\cA^z  (\sigma_\mB^{-z} \rho_\mB^{z} \otimes 1_{\mathcal C}) \rho_\cA^{1/2-z} \right> 
\end{align}
and
\begin{equation}
\Delta_{\psi}^{1/2-z} \vphantom{\sum} \left| \Gamma_\psi(z) \right> =  \left| \rho_\cA^{1/2-z} \sigma_\cA^z  (\sigma_\mB^{-z} \rho_\mB^{z} \otimes 1_{\mathcal C})  \right> .
\end{equation}
The $L_p(\cA',\psi)$ norms can be computed using the well known correspondence between these norms and the sandwiched relative entropy discussed in \cite{Berta2}.
This gives:
\begin{equation}
\left\| \left| \Gamma_\psi(\theta) \right>  \right\|_{p, \psi}^{\mathcal{\cA}'} = \left({\rm tr}\left| \rho_\cA^{1/p-1/2} \Gamma_\psi(\theta) \right|^p \right)^{1/p}
=  \left({\rm tr}\left| \rho_\cA^{\theta}\sigma_\cA^\theta  (\sigma_\mB^{-\theta} \rho_\mB^{\theta} \otimes 1_{\mathcal C}) \rho_\cA^{1/2-\theta} \right|^{p_\theta} \right)^{1/p_\theta},
\end{equation}
where in the last equation we set $p=p_\theta$ and used $1/p_\theta -1/2 = \theta$, and where $|\psi\rangle = \left| \right. \rho_\cA^{1/2} \left. \right>$. Similarly, we have
\begin{equation}\label{finitergred}
\left\|\Delta_{\psi}^{1/2-\theta} \left| \Gamma_\psi(\theta) \right>  \right\|_{p_\theta,\psi}^{\mathcal{A}'} 
=  \left({\rm tr}\left| \rho_\cA^{1/2}\sigma_\cA^\theta  (\sigma_\mB^{-\theta} \rho_\mB^{\theta} \otimes 1_{\mathcal C}) \right|^{p_\theta} \right)^{1/p_\theta}
\end{equation}
and we recognize this later expression as  \cite{Junge}, eq. (25) with $\alpha$ there given by $p_\theta/2$. 

\subsection{Example: half-sided modular inclusions}

Half-sided modular inclusions were introduced in \cite{Wiesbrock1,Wiesbrock2} and consist of the following data:
An inclusion $\cB \subset \cA$ of v. Neumann algebras acting on a common Hilbert space $\sH$, containing a 
common cyclic and separating vector $|\eta\rangle$. Furthermore, for $t \ge 0$, 
it is required that $\Delta_{\eta, \cA}^{it} \cB \Delta^{-it}_{\eta, \cA} \subset \cB$, hence the terminology ``half-sided.''
This situation is common for light ray algebras in chiral CFTs, where $|\eta\rangle$ is the vacuum.

 Wiesbrock's theorem \cite{Wiesbrock1,Wiesbrock2} is the result that for any half-sided modular inclusion, there exists a 1-parameter unitary group $U(s), s \in \mathbb R$ with self-adjoint, non-negative 
 generator which can be normalized so that 
 \ben
\Delta_{\eta, \cA}^{-it} \Delta_{\eta, \cB}^{it} = U(e^{2\pi t}-1)
\een
for $t \in \mathbb R$.
Furthermore, the unitaries $\Delta_{\eta,\cA}^{it}, U(s)$ fulfill the Borchers commutation relations \cite{Borchers} and 
in particular $\cB = U(1) \cA U(1)^*$, $J_{\cA} U(s) J_{\cA} = U(-s)$.  For any $a>0$, the inclusion 
$\cA_a = U(a) \cA U(a)^* \subset \cA$ is then also half sided modular. 

For a half-sided modular inclusion, the embedding is trivial, $V_\eta=1$. Using this information, one can easily show 
that in the case of the half-sided modular inclusions $\cA_a = U(a) \cA U(a)^* \subset \cA$, the rotated Petz recovery channel, denoted here as $\alpha_a^t: \cB \to \cA$ to emphasize the dependence on $a$, is:
\begin{equation}
\alpha_\eta^t(x) \equiv U(a(1+ e^{-2\pi t}))^* x U(a(1+ e^{-2\pi t})).
\end{equation}
Thm. \ref{thm1} therefore gives the following corollary, conjectured in \cite{Faulkner}, after a change of 
integration variable.
\begin{corollary}
Let $\cB \subset \cA$ be a half-sided modular inclusion with respect to the
reference vector $|\eta\rangle$, so $\cB=\cA_a=U(a) \cA U(a)^*$. Then we have
\ben
 \frac{1}{a}[S_{\cA}(\omega_\psi| \omega_\eta ) -  S_{\cA_a}(\omega_\psi| \omega_\eta )]\geq
  \int_{a}^{\infty} \ln F( \omega_\psi, U(y) \omega_\psi U(y)^* )^2  \, \frac{\ud y}{y^2} .
\een
\end{corollary}

 For a half-sided modular inclusion, $V_\psi = u_{\psi;\eta}' \in \mathcal{B}'$ [from \eqref{Vchi}] is the partial isometry that takes $|\psi_\cA \rangle$ in the natural cone $\sP^\natural_{\cA}$ (defined w.r.t. $|\eta\rangle$) to the state representer in 
 $\sP^\natural_{\cB}$ (also defined w.r.t. $|\eta\rangle$).
The interpolation vector \eqref{ourvec} thereby becomes in the case of half sided modular inclusions
\begin{equation}
\label{hmvec}
\left| \Gamma_\psi(z) \right> 
=\Delta_{\eta_\cA,\psi_\cA}^z \Delta_{\eta_\mB,\psi_\mB}^{-z} \left| \psi \right>. 
\end{equation}
The vector \eqref{hmvec} is similar to a vector studied in \cite{Faulkner} in order to prove the quantum null energy condition (QNEC). Based on this and some preliminary calculations we speculate here that the QNEC can be understood in terms of the strengthened monotonicity result in Theorem~\ref{thm1}.
\begin{conjecture}
The limit $a \rightarrow 0$ of thm. \ref{thm1} in the case of a half-sided modular 
inclusion $\cA_a = U(a) \cA U(a)^* \subset \cA$ leads to a saturation of the bound:
\begin{equation}
\lim_{a \rightarrow 0}  \ \frac{2}{a} \int_{-\infty}^{\infty}  \ln F( \rho, \rho \circ \alpha_a^t ) p(t) \ud t =  \frac{\ud}{\ud a} S_{\mathcal{A}_a}(\rho|\sigma) \bigg|_{a=0}.
\end{equation}
\end{conjecture}
This is a more refined version of a conjecture appearing in \cite{Faulkner}.
A corollary to this conjecture, if proven, would be a new proof of the QNEC since the recovery channel is translationally invariant so applying the same result to a further translated null cut
one can use monotonicity of the fidelity to prove that $\frac{\ud}{\ud a} S_{\mathcal{A}_a}(\rho|\sigma)$ is monotonic in $a$ as required by the QNEC. 

\vspace{1cm}

{\bf Acknowledgements:} SH\ is grateful to the Max-Planck Society for supporting the collaboration between MPI-MiS and Leipzig U., grant Proj.~Bez.\ M.FE.A.MATN0003. 
TF and SH benefited from the KITP program ``Gravitational Holography''. This research was supported in part by the National Science Foundation under Grant No. NSF PHY-1748958. BGS and YW acknowledge that this material is based in part on work supported by the Simons Foundation as part of the It From Qubit Collaboration and in part on work supported by the Air Force Office of Scientific Research under award number FA9550-19-1-0360. 
YW would like to acknowledge discussions with Jonathan Rosenberg.  TF acknowledges part of the work presented here is support by the DOE under grant DE-SC0019517.

\appendix

\section{Proof of lemma \ref{lem:2}} \label{lemma2}

\begin{proof}
We apply Zorn's lemma. Consider the following set of projectors:
\begin{equation}
\Pi_{\sS} = \{ \pi(\rho_i) -  \pi(\rho_j) : \rho_{i,j} \in \sS\,, \,\, \pi(\rho_j) \leq \pi(\rho_i) \}
\end{equation}
where the later condition requires a proper subset. These differences are still projectors since $(\pi(\rho_i) -  \pi(\rho_j))^2 = \pi(\rho_i) - \pi(\rho_j)$ by the inclusion condition which implies that $\pi(\rho_j) \pi(\rho_i) = \pi(\rho_i)$. 

If $\Pi_{\sS}$ is the empty set then it must be the case that $\pi(\rho_i) = \pi(\rho_j)$ for all $\rho_{i,j} \in \sS$, since otherwise we could use convexity to show a contradiction:
\begin{equation}
\pi\left(\frac{ \rho_i + \rho_j}{2} \right) - \pi(\rho_i) \in \Pi_{\sS}.
\end{equation}
So in this case \eqref{allincl} is trivial. 

We may thus assume from now on that $\Pi_{\sS}$ is non-empty.   By Zorn's lemma we can pick a maximal family of mutually orthogonal projectors from $\Pi_{\sS}$, where family means a subset of $\Pi_{\sS}$, and maximal means that there are no other orthogonal families of projectors that are strictly larger under the order of inclusion. Call the maximal family $q_{\rm max}$. By the $\sigma$-finite condition, it is a countable family
\begin{equation}
q_{\rm max} = \{ \pi(\rho_{i_n}) -  \pi(\rho_{j_n}):\,\, n=1,2\ldots \}.
\end{equation}
Given $q_{\rm max}$ we define:
\begin{equation}
\sigma = \sum_{n=1}^{\infty} 2^{-n} \rho_{i_n}
\end{equation}
The infinite sum converges in the linear functional norm and so by convexity and closedness of $\sS$ we find that $\sigma \in \sS$. 
The support projector for this state satisfies:
\begin{equation}
\pi(\sigma) \mathscr{H} = \bigoplus_{n=1}^\infty \pi(\rho_{i_n}) \mathscr{H}
\end{equation}
(understood as a direct sum in the norm topology.) 
By the maximality condition we can show \eqref{allincl}. To see this, suppose that this is not true for some $\rho_k$ then:
\begin{equation}
\mathcal{B}_\sigma \subset \pi\left(\frac{\sigma+\rho_k}{2}\right) - \pi(\sigma) \in \Pi_{\sS} \quad \text{and} \quad \left( \pi\left(\frac{\sigma+\rho_k}{2}\right) - \pi(\sigma) \right) \perp \left( \pi(\rho_{i_n}) -  \pi(\rho_{j_n}) \right) 
\end{equation}
for all $n$. This contradicts the maximality of $q_{\rm max}$, which is absurd. 
\end{proof}

\section{Isometric embedding}
\label{app:Vsigma}

We work with $\sigma \in \mathcal{A}_\star$ faithful which implies that $\sigma \circ \iota \in {\mathcal B}_\star$ 
is faithful. Thus the corresponding vectors $|\xi_{\sigma}^\cA \rangle, |\xi_{\sigma}^{\mathcal{B}} \rangle$ 
in the natural cones are cyclic and separating. By a trivial calculation, one sees that $V_\sigma$ defined in \eqref{embed} is a norm-preserving (densely defined) map from $\mathscr{K}$ to $\mathscr{H}$. So the map extends to the full Hilbert space as an isometric embedding $V_\sigma^* V_\sigma = 1_{\mathscr{K}}$. A similar argument shows that:
\begin{equation}
 \qquad V_\sigma V_\sigma^* = \pi_{\mathscr{K}}  \in B(\mathscr{H})
\end{equation}
where this equation applies on the subspace of $\mathscr{H}$ that is generated by $\mathcal{B}$:
\begin{equation}
\overline{ \iota(\mathcal{B} )\left| \xi_\sigma^\cA \right>} = \pi_{\mathscr{K}} \mathscr{H} \equiv \pi^{\mathcal{B}'}( \sigma)  \mathscr{H}
\end{equation}
In other words, $|\xi_\sigma^\cA \rangle$ is not cyclic for $\iota(\mathcal{B})$ and $\pi^{\mathcal{B}'}( \sigma)$ defines the associated support projector for the commutant algebra.

The embedding satisfies:
\begin{equation}
V_\sigma  b \left| \chi \right> = b V_\sigma \left| \chi \right> \,, \qquad \chi \in \mathscr{K} \,, \qquad b \in \mathcal{B}
\end{equation}
since we can approximate any $\left| \chi \right> = \lim_n c_n \left| \xi_{\sigma}^{\mathcal{B}} \right> \in \mathscr{K}$ for suitable $c_n \in \mathcal B$, and take the limit on both sides of:
\begin{equation}
V_\sigma b c_n \big| \xi_{\sigma}^{\mathcal{B}}\big>  = \iota(b c_n) \left| \xi_\sigma\right> = \iota(b) \iota(c_n) \left| \xi_\sigma\right>
= \iota(b) V_\sigma c_n \big| \xi_{\sigma}^{\mathcal{B}} \big>.
\end{equation}
Thus,
\begin{equation}
\left< \chi_1 \right| V_\sigma^* \iota(b) V_\sigma \left| \chi_2\right>
= \left<\chi_1 \right| V_\sigma^*  V_\sigma b \left| \chi_2 \right> 
= \left< \chi_1 \right|  b \left| \chi_2 \right> 
\end{equation}
for all vectors $|\chi_{1,2} \rangle \in\mathscr{K}$, or:
\begin{equation}
 V_\sigma^* \iota(\mathcal{B}) V_\sigma^{} = \mathcal{B}.
\end{equation}

The commutant satisfies:
\begin{equation}
\label{VAprime}
V_\sigma^* \mathcal{A}' V_\sigma \subset \mathcal{B}'
\end{equation}
which can be verified via a short calculation for $a' \in \mathcal{A}'$ and $b \in \mathcal{B}$:
\begin{align}
\left[ V_\sigma^* a' V_\sigma, b \right] = \left[ V_\sigma^* a' V_\sigma, V_\sigma^* \iota(b) V_\sigma \right] 
= V_\sigma^* \left[ \pi_{\mathscr{K}} a'  \pi_{\mathscr{K}}, \iota(b) \right] V_\sigma
= 0
\end{align}
where we used the fact that $\pi_{\mathscr{K}} \in \iota(\mathcal{B})'$ and $\mathcal{A}' \subset \iota(\mathcal{B})'$. 

\section{Fidelity}

\subsection{Proof of Lemma~\ref{lem:fid} (Fidelity and the Araki-Masuda norm)}
\label{app:lem:fid}

\begin{proof}
(1) In this proof, all $L_1$ norms are taken relative to the commutant $\cA'$ as in 
\begin{equation}
\label{AMdefapp}
\left\| \phi \right\|_{1,\psi}
= \inf_{ \chi \in \mathscr{H}:\|\chi  \| = 1, \pi'(\chi) \geq \pi'(\phi) }
\| ( \Delta_{\chi,\psi}')^{-1/2} \phi \|,
\end{equation}
from \eqref{AMdef},
and we want to relate this to the fidelity,
\begin{equation}
\label{uhl2app}
F(\omega_\psi,\omega_\phi) = \sup_{x' \in \mathcal{A}': \| x\| \leq 1} |\left< \psi \right| x' \left| \phi \right>|
\end{equation}
where  $\phi,\psi$ are normalized vectors. 
This relation is proven in \cite{AM}, lem. 5.3 for a cyclic and separating vector $|\psi\rangle$. We will now remove this condition.
The linear functional that appears in \eqref{uhl2app} $\mathcal{A}'$ can be written using a polar decomposition 
\begin{equation}
 \left<  \psi \right| \cdot \left| \phi \right> = \left< \xi \right| \,\cdot\,\, u' \left| \xi \right>
\end{equation}
for some $\xi$ in the natural cone and a partial isometry $u'$ with initial support $(u')^* u' = \pi'(\xi)$.
This polar decomposition has the property that the largest projector in $\mathcal{A}'$ that satisfies $ \left< \xi \right| x' p' u' \left| \xi \right> =0 $ for all $x'$ is $p' = 1- \pi'(u' \xi) = 1 - u' (u')^*$.\footnote{Proof: Certainly $1- u' (u')^*$ satisfies this. Suppose $p'$ is larger and still satisfies this. Pick $x' = (u')^*$, then $\left< \xi \right| (u')^* p' u' \left| \xi \right> = 0$, but then $p' \leq  1- \pi'( u' \xi)$ which is a contradiction. Note that the largest projector in $\mathcal{A}'$ that satisfies
$ \left< \xi \right| p' x' u' \left| \xi \right> =0 $ for all $x'$ is $p' = 1- \pi'( \xi) = 1 - (u')^* u'$.} 
Thus:
\begin{equation}
\left< \psi \right| x' (1 - u' (u')^*) \left| \phi \right> = 0\, , \quad \forall x' \in \mathcal{A}'
\end{equation}
and since $\overline{ \mathcal{A}' \left| \psi \right> } = \pi(\psi) \mathscr{H}$ we derive that
the final support projector satisfies:
\begin{equation}
( 1 - u' (u')^* ) \left| \phi \right>  \in (1- \pi(\psi) ) \mathscr{H}
\label{upup}
\end{equation}
Consider
\begin{align}
&\left((x')^* \left| \psi \right>,  (u')^* \left| \phi \right> \right)
= \left( \left| \xi \right>,  \pi'(\psi)  x' \pi'(\xi) \left| \xi \right> \right) = \left( \left| \xi \right>,  \pi'(\psi) x'  \left| \xi \right> \right) 
\nonumber \\
& = \left( J (\Delta_{\xi,\psi}')^{1/2} \left| \psi \right>, J (\Delta_{\xi,\psi}')^{1/2} (x')^* \left| \psi \right> \right) \nonumber \\
& =  \left( (\Delta_{\xi,\psi}')^{1/2} (x')^* \left| \psi \right> , (\Delta_{\xi,\psi}')^{1/2} \left| \psi \right>\right) 
\end{align}
where in the second line we used \eqref{modA} and in the third we used the anti-unitarity of $J$. 
The above relation can be rewritten as:
\begin{equation}
\label{deqeq}
\left((x')^* \left| \psi \right> + \left| \zeta \right>,  \pi(\psi) (u')^* \left| \phi \right> \right) = 
 \left( (\Delta_{\xi,\psi}')^{1/2} ((x')^* \left| \psi \right> +\left| \zeta \right>), (\Delta_{\xi,\psi}')^{1/2} \left| \psi \right>\right) 
\end{equation}
where we have freely added $\zeta \in (1 - \pi(\psi) \pi'(\xi) ) \mathscr{H}$ since $\pi(\psi) (u')^* \left| \phi \right>$ is in the subspace $  \pi(\psi) \pi'(\xi)  \mathscr{H}$,
and this subspace is also the support of $\Delta_{\xi,\psi}'$. Now since the vector on the left of \eqref{deqeq} is dense:
$\pi'(\xi) \overline{\mathcal{A} \left| \psi \right>} + ( 1-  \pi(\psi) \pi'(\xi) ) \mathscr{H} = \mathscr{H}$ we learn that $(\Delta_{\xi,\psi}')^{1/2} \left| \psi \right>$
is in the domain of $(\Delta_{\xi,\psi}')^{1/2}$ and
\begin{equation}
\Delta_{\xi,\psi}' \left| \psi \right> = \pi(\psi) (u')^* \left| \phi \right>,
\end{equation}
so that
\begin{equation}
u' \Delta_{\xi,\psi}' \left| \psi \right> = \pi(\psi) u' (u')^* \left| \phi \right>  =  \pi(\psi)  \left| \phi \right>,
\end{equation}
where we used \eqref{upup}.
The next step is to show that
\begin{equation}
\label{ohno}
\|  \phi \|_{1,\psi} =  \| \pi(\psi) \phi  \|_{1,\psi} = 
 \| u' \Delta_{\xi,\psi}'  \psi \|_{1,\psi}  = \| \xi \|^2 ,
\end{equation}
which implies that
\begin{equation}
\| \phi  \|_{1,\psi} = \sup_{x' \in \mathcal{A}': \|x' \| \leq 1}| \left< \xi \right| x' \left| \xi \right> | =  \sup_{x' \in \mathcal{A}': \| x' \| \leq 1} | \left<  \psi \right| x' \left| \phi \right> | = F(\omega_\psi,\omega_\phi) .
\end{equation}
This is what we wanted to derive. 

The later equality in \eqref{ohno} is fairly standard, but for completeness we go through this. Without loss of generality we take $\chi$ in \eqref{AMdefapp} such that $u' \Delta_{\xi,\psi}' \left| \psi \right> $
is in the domain of $(\Delta_{\chi,\psi}')^{-1/2}$ and also such that $\pi'(\chi) \geq \pi'(u' \Delta_{\xi,\psi}' \left| \psi \right>) = \pi'( u' \xi) $
and $\| \left| \chi \right> \| = 1$. 
We would like to use the following result that we will justify later (for now the reader should feel free to verify this for type-I algebras with density matrices):
\begin{equation}
\label{toprove}
(\Delta_{\chi,\psi}')^{-1/2} u' \Delta_{\xi,\psi}' \left| \psi \right> =  (\Delta_{\chi,\xi}')^{-1/2} u' \left| \xi \right> 
= (\Delta_{\chi,\xi}')^{-1} j(u')^* \left| \chi \right> 
\end{equation}
where $j(u') = J u' J$ and all the domains in the above equation are appropriate. Now apply the Cauchy-Schwarz inequality:
\begin{align}
\|  \xi \|^2 &= \| \pi'(\chi) u' \xi \|^2 =  \|  \Delta_{\xi,\chi}^{1/2} j(u')^* \chi  \|^2  = \left< \chi \right| j(u')  (\Delta_{\chi,\xi}')^{-1} j(u')^* \left| \chi \right>\\
&\leq \| (\Delta_{\chi,\xi}')^{-1} j(u')^* \chi  \|  \|  j(u')^* \chi \| \leq 
\| (\Delta_{\chi,\psi}')^{-1/2} u' \Delta_{\xi,\psi}' \psi \|.
\end{align}
Taking the infimum over all such $\chi$ we find that:
\begin{equation}
\| \xi \|^2 \leq \| u' \Delta_{\xi,\psi}' \psi \|_{1,\psi}= \| \phi \|_{1,\psi}.
\end{equation}
The other inequality is found since the optimal vector in the infimum is $\left|\chi\right> = u'\left|\xi \right>/\| \left| \xi \right> \|$ where \eqref{toprove} becomes:
\begin{equation}
(\Delta_{\chi,\psi}')^{-1/2} u' \Delta_{\xi,\psi}' \left| \psi \right> 
= u' \left| \xi \right> \| \xi \|
\end{equation}
which implies that:
\begin{equation}
 \|  \phi \|_{1,\psi}  = \| u' \Delta_{\xi,\psi}' \psi \|_{1,\psi}\geq \| \xi \|^2 
\end{equation}
and this establishes equality. We now only need to prove \eqref{toprove}. 
To do this we will analytically continue the equation:
\begin{equation}
\label{aboveeq}
 (\Delta_{\chi,\psi}')^{-z} u' (\Delta_{\xi,\psi}')^{z} \left| \xi \right> =  \pi(\psi)  (\Delta_{\chi,\xi}')^{-z} u' (\Delta_{\xi,\xi}')^{z} \left| \xi \right> = 
 \pi(\psi)   (\Delta_{\chi,\xi}')^{-z} u' \left| \xi \right>
\end{equation}
away from $z= is$ for $s$ real. 
We simply take an inner product with a dense set of vectors $a \left| \chi \right> + \left| \zeta \right>$ where $a \in \mathcal{A}$ and $\left| \zeta \right> \in (1 - \pi'(\chi)) \mathscr{H}$:
\begin{equation}
\left(  (\Delta_{\chi,\psi}')^{-\bar z} (a \left| \chi \right> + \left| \zeta \right>) ,  u' (\Delta_{\xi,\psi}')^{z} \left| \xi \right> \right) = \left(  (\Delta_{\chi,\xi}')^{-\bar z}\pi(\psi)  (  a \left| \chi \right> +   \left| \zeta \right>) , u'  \left| \xi \right> \right)
\end{equation}
since we know that $ \left| \xi \right>$ is in the domain of $(\Delta_{\xi,\psi}')^{1/2}$ (since we established that $ \left| \psi \right>$ is in the domain of $\Delta_{\xi,\psi}' $)
it is clear that we can analytically continue the two functions above into the strip $0 < {\rm Re} z < 1/2$ with continuity in the closure (using standard results
in Tomita-Takesaki theory.) Agreement along $z=is$ implies agreement in the full strip. Setting $z=1/2$ we have a uniform bound (with $ \| a \left| \chi \right>\| \leq 1$) on the left hand side since we started with the assumption that the left hand side of \eqref{toprove} exists. On the right hand side this establishes the fact that $u'  \left| \xi \right>$ is in the domain $ (\Delta_{\chi,\xi}')^{-1/2}$ and the first equality in \eqref{toprove}. The second equality in \eqref{toprove} is immediate. 

We have thus finished the proof that \eqref{AMdefapp} and \eqref{uhl2app} are equal. 

\vspace{12pt} 

(2) Our next task is to show that it is sufficient to vary over unitaries in \eqref{uhl2app} and relate this to \eqref{uhl1}. 
Note that for a bounded operator we have the polar decomposition $ x' = u' p'$ where $u'$ is unitary and $\| p' \| \leq 1$. Such a self adjoint operator can always be written as $(v' + (v')^*)/2$ where $v' = p' + i \sqrt{ 1- (p')^2}$. So:
\begin{equation}
x' = \frac{1}{2} u' v' + \frac{1}{2} u' (v')^* =  \frac{1}{2} w' + \frac{1}{2} y'
\end{equation}
for unitaries $w', y' \in \mathcal{A}'$. Then $ | \left<  \psi \right| x' \left| \phi \right> | \leq \frac{1}{2} \left( | \left<  \psi \right| w' \left| \phi \right> |+ | \left<  \psi \right| y' \left| \phi \right> | \right)$. Thus 
\begin{equation}
| \left<  \psi \right| x' \left| \phi \right> | \leq \sup_{u' \in \mathcal{A}': u' (u')^* =1} | \left<  \psi \right| u' \left| \phi \right> |
\end{equation}
since the right hand side is larger than both terms with $w'$ and $y'$ above. Taking the sup over the left hand side:
\begin{equation}
 \sup_{u' \in \mathcal{A}': u' (u')^* \leq 1} | \left<  \psi \right| u' \left| \phi \right> |  \leq \sup_{x' \in \mathcal{A}': \| x' \| \leq 1} | \left<  \psi \right| x' \left| \phi \right> | \leq \sup_{u' \in \mathcal{A}': u' (u')^* =1} | \left<  \psi \right| u' \left| \phi \right> |
\end{equation}
where the first inequalities is because the set of unitaries is a subset of operators bounded by $1$. This implies equality and we see that the $L_1$
norm is equivalent to the Uhlmann fidelity of two linear functionals:
\begin{equation}
F(\omega_\psi, \omega_{\phi}) = \|  \phi \|_{1,\psi}  \,, \quad \omega_\psi = \left< \psi \right| \cdot \left| \psi \right> \,, \quad \omega_{\phi} = \left< \phi \right| \cdot \left| \phi \right> \in \mathcal{A}_\star.
\end{equation} 
Thus it is clear the fidelity is independent of the vector representation.
We take the norms of $\phi,\psi$ to be $1$.

\vspace{12pt} 

(3) Finally, we want to relate the fidelity to the norm of the linear functional difference:
\begin{equation}
 \| \omega_\psi - \omega_\phi \|\equiv \sup_{x \in \mathcal{A}; \| x\| \leq 1} | \omega_\psi (x)  - \omega_\phi(x) |
\end{equation}

Since $\mathcal{A} \subset B(\mathscr{H})$:
\begin{equation}
 \| \omega_\psi - \omega_\phi \| \leq \sup_{x \in B(\mathscr{H}); \|x\| \leq 1} |  \left< \psi \right| x \left| \psi \right>  -  \left< \phi \right| (u')^* x u' \left| \phi \right> |=
2\sqrt{  1 - \left| \big< \psi \big| u' \phi \big> \right|^2 }
\label{linfun}
\end{equation}
We calculate the last equality as follows. The two normalized vectors $\psi, u' \phi$ live in a two dimensional subspace, which without loss of generality can be chosen as:
\begin{equation}
\left| \psi \right> = \cos(\theta/2) \left| 0 \right> + \sin(\theta/2) \left| 1 \right> \,, \qquad u' \big| \phi \big> = e^{i\varphi}(
\sin(\theta/2) \left| 0 \right> + \cos(\theta/2) \left| 1 \right>)
\end{equation}
where $\big< \psi \big| u' \phi\big> =e^{i\varphi}\sin(\theta) $.
We can then take $x$ to be an operator in this subspace. Note that:
\begin{equation}
\left| \psi \right> \left< \psi \right|  - \big| u' \phi\big> \big< u'\phi \big| = \cos\theta \sigma_3
\end{equation}
such that the maximum is achieved for $x = \sigma_3={\rm diag}(1,-1)$ which has an operator norm of $1$. So the norm of this linear functional
is $2 \cos\theta$, giving the last equality in \eqref{linfun}. Taking the $\inf$ over $u'$ in \eqref{linfun}, we have:
\begin{equation}
 \| \omega_\psi - \omega_\phi \| \leq 2  \sqrt{ 1 - F(\omega_\psi, \omega_{\phi})^2}.
\end{equation}

In the other direction we can pick $\phi$ and $\psi$ to live in the natural cone without loss of generality, and then we have
$
\| \omega_\psi - \omega_\phi \| \geq | \left| \phi \right> - \left| \psi \right>|^2 = 2(1 -  \left< \psi \right| \left. \phi \right>)
$
where the later quantity is real since both vectors are in the cone. 
We use the inequality \eqref{petzAM} for $p=1$ that we reproduce here:
\begin{equation}
\left\|   \phi \right\|_{1,\psi} \geq \left< \phi \right| (\Delta_{\psi,\phi}')^{1/2} \left| \phi \right> 
=  \left< \phi \right| \left. \psi \right>, 
\end{equation}
so
\begin{equation}
\frac{1}{2} \| \omega_\psi - \omega_\phi \| \geq  1-  F(\omega_\psi,\omega_\phi)
\end{equation}
Altogether, we have
\begin{equation}
1 -  F(\omega_\psi,\omega_\phi) \leq  \frac{1}{2} \| \omega_\psi - \omega_\phi \| \leq   \sqrt{  1 - F(\omega_\psi,\omega_\phi)^2 }.
\end{equation}
Note that the fidelity lies between $0$ and $1$ and:
\begin{equation}
0\leq  \| \omega_\psi - \omega_\phi \| \leq 2
\end{equation}
where equality is achieved on the left iff the two linear functionals are the same and on the right if the support of the two linear functionals are orthogonal.
We can see this as follows. Note that for $\|x\| \leq 1$:
\begin{equation}
| \omega_\psi(x) | \leq \| x\| \omega_\psi(1) \leq 1
\end{equation}
so that $| \omega_\psi(x) - \omega_{\phi}(x)|$ lies between $0$ and $2$. Equality is achieved for $x = \pi(\psi) - \pi(\phi)$ with orthogonal support. 
\end{proof}

\subsection{Proof of Lemma~\ref{lem:Fcont} (Continuity of fidelity)}
\label{sec:lem:p1}

In this section, all $L_1$ norms refer to the commutant algebra $\cA'$, as in 
(lem.~\ref{lem:fid}):
\begin{equation}
\label{eq:variation}
 \left\| \psi \right\|_{1,\phi}  = F(\omega_{\psi}, \omega_{\phi} ) \equiv \sup_{ u' \in \mathcal{A}'} \left| \left< \psi \right| u' \left| \phi \right> \right|,
\end{equation}
where the supremum is over partial isometries $u'$.

\begin{lemma}
\label{lem:Fcont}
For a v. Neumann algebra $\cA$ in standard form acting on a Hilbert space $\sH$ and any $\ket{\psi_i}, \ket{\phi_i} \in \sH$, 
\ben
\left|  \, \left\| \psi_1 \right\|_{1,\phi_1} -
\left\| \psi_2 \right\|_{1,\phi_2} \right|
 \leq  \left\| \phi_1-\phi_2 \right\| +  \left\| \psi_1 -\psi_2 \right\|.
\een
\end{lemma}
\begin{proof} 
The variational expression \eqref{eq:variation} immediately allows one to deduce the triangle inequality for the $L_1$-norms. Note that:
\begin{equation}
\label{bdp}
\sup_{ u' \in \mathcal{A}'} \left| \left< \psi_1 \right| u' \left| \psi_2 \right> \right|
\leq \sup_{ u \in B(\mathscr{H})} \left| \left< \psi_1 \right| u \left| \psi_2 \right> \right|
= \left\| \psi_1 \right\| \left\| \psi_2\right\|, 
\end{equation}
and that $ \left\| \psi \right\|_{1,\phi} =  \left\| \phi \right\|_{1,\psi}$ are further trivial consequences of the variational definition.
For normalized vectors $\psi_1,\psi_2, \phi_1,\phi_2$ we derive for the $L_1$-norms relative to $\cA'$:
\begin{align}
\left|  \, \left\| \psi_1 \right\|_{1,\phi_1} -
\left\| \psi_2 \right\|_{1,\phi_2} \right|
&\leq 
\left| \, \left\| \psi_1 \right\|_{1,\phi_1} -
\left\| \psi_1 \right\|_{1,\phi_2} \right|
+
\left| \, \left\| \psi_1 \right\|_{1,\phi_2} -
\left\| \psi_2 \right\|_{1,\phi_2} \right|
\nonumber\\
&\leq
\| \phi_1 -\phi_2\|_{1,\psi_1}
+
\| \psi_1 -\psi_2\|_{1,\phi_2} \leq  \left\| \phi_1-\phi_2 \right\| +  \left\| \psi_1 -\psi_2 \right\|
\label{contboth}
\end{align}
where to go to  the second line we used the reverse triangle inequality twice,
and in the last step we used \eqref{bdp}. 
\end{proof}

\section{Proof of lemma~\ref{lem:hirsch} (Hirschman's improvement)}
\label{sec:lem:hirsch}

\begin{proof}
(1) First assume that $\omega_\psi$ is faithful and we may assume $|\psi\rangle \in \mathscr{P}^\natural_{\mathcal M}$
by invariance of the $L_p$-norms. Then $|\psi\rangle$ is cyclic and separating and the standard theory 
developed in \cite{AM} applies. We use the notation $\bS_{1/2}=\{0< {\rm Re} z < 1/2\}$. 

Denote the dual of a H\" older index $p$ by $p'$, 
defined so that $1/p+1/p'=1$. \cite{AM} have shown that the non-commutative $L^p(\mathcal{M}, \psi)$-norm of a vector $|\zeta\rangle$ relative to $\ket{\psi}$
can be characterized by (dropping the superscript on the norm)
\ben
\| \zeta \|_{p,\psi} = \sup \{ |\langle \zeta | \zeta' \rangle | : \ \ \| \zeta' \|_{p',\psi} \le 1 \}. 
\een
They have furthermore shown that when $p'\ge 2$, any vector $|\zeta'\rangle \in L^{p'}(\mathcal{M},\psi)$ has a unique generalized polar decomposition, i.e. can be written in 
the form $|\zeta'\rangle = u \Delta_{\phi,\psi}^{1/p'} |\psi\rangle$, where $u$ is a unitary or partial isometry from $\mathcal{M}$. Furthermore, they show that $\| \zeta' \|_{p',\psi} = \|\phi\|^{p'}$. We may thus choose a $u$ and a normalized $|\phi \rangle$, so that 
\ben
 \| G(\theta) \|_{p(\theta),\psi} = \langle u \Delta_{\phi,\psi}^{1/p(\theta)'} \psi | G(\theta) \rangle, 
\een
perhaps up to a small error which we can let go zero in the end. Now we define $p_\theta$ as in the statement, so that 
\ben
\frac{1}{p_\theta'} = \frac{1-2\theta}{p_0'} + \frac{2\theta}{p_1'}, 
\een
and we define an auxiliary function $f(z)$ by 
\ben\label{eq:fdef}
f(z)=\langle u \Delta_{\phi,\psi}^{2\bar z/p_1'+(1-2\bar z)/p_0'} \psi | G(z) \rangle, 
\een
noting that 
\ben
f(\theta)= \| G(\theta) \|_{p_\theta,\psi} 
\een
by construction.
By Tomita-Takesaki-theory, $f(z)$ is holomorphic in $\bS_{1/2}$. For the values at the boundary of the strip $\bS_{1/2}$, we estimate
\ben\label{eq:est1}
\begin{split}
|f(it)| & = |\langle u \Delta_{\phi,\psi}^{-2it(1/p_1'-1/p_0')} \Delta_{\phi,\psi}^{1/p_0'} \psi | G(it) \rangle|  \\
        & \le \| u \Delta_{\phi,\psi}^{-2it(1/p_1'-1/p_0')} \Delta_{\phi,\psi}^{1/p_0'} \psi \|_{p_0',\psi} \| G(it) \|_{p_0,\psi}  \\
        & \le \| \Delta_{\phi,\psi}^{-2it(1/p_1'-1/p_0')} \Delta_{\phi,\psi}^{1/p_0'} \psi \|_{p_0',\psi} \| G(it) \|_{p_0,\psi}  \\
        & \le \|\phi\|^{p_0'} \| G(it) \|_{p_0,\psi}  \\
        & \le \| G(it) \|_{p_0,\psi}  .
\end{split}
\een
Here we used the version of H\" older's inequality proved by \cite{AM}, we used $ \| a^* \zeta \|_{p_0',\psi} \le \|a\| \| \zeta \|_{p_0',\psi}$ for 
any $a \in \cA$, see \cite{AM}, lem. 4.4, and we used $\| \Delta_{\phi,\psi}^{-2it(1/p_1'-1/p_0')} \Delta_{\phi,\psi}^{1/p_0'} \psi \|_{p_0',\psi} \le \|\phi\|^{p_0'}$ which we 
prove momentarily. A similar chain of inequalities also gives 
\ben
\label{eq:est2}
|f(1/2+it)| \le \| G(1/2+it) \|_{p_1,\psi}.
\een
To prove the remaining claim, let $|\zeta'\rangle = \Delta^{z}_{\phi,\psi}|\psi\rangle$ and $z=1/p'+2it$. Then we have, using the variational characterization 
by \cite{AM} of the $L^{p'}(\mathcal{M},\psi)$-norm when $p'\ge 2$:
\ben
\begin{split}
\|\zeta'\|_{p',\psi} =& \sup_{\| \chi \|=1} \| \Delta_{\chi,\psi}^{1/2-1/p'} \Delta^{z}_{\phi,\psi} \psi  \| \\
=& \sup_{\| \chi \|=1} \| \Delta_{\chi,\psi}^{1/2-1/p'-2it} \Delta^{1/p'+2it}_{\phi,\psi}  \psi  \| \\
=& \sup_{\| \chi \|=1} \| \Delta_{\chi,\psi}^{1/2-1/p'} (D\chi:D\phi)_{2t} \pi^{\mathcal M}(\phi) \Delta^{1/p'}_{\phi,\psi}  \psi \| \\
\le & \sup_{\| \chi \|=1, a \in \cA, \|a\|=1} \| \Delta_{\chi,\psi}^{1/2-1/p'} a \Delta^{1/p'}_{\phi,\psi}  \psi \| \\
\le & \sup_{a \in \cA, \|a\|=1} \| a \Delta^{1/p'}_{\phi,\psi}  \psi  \|_{p',\psi}. 
\end{split}
\een
Using \cite{AM}, lem. 4.4, we continue this estimation as 
\ben
\le  
\sup_{a \in \cA, \|a\|=1} \| a \|   \|\Delta^{1/p'}_{\phi,\psi}  \psi  \|_{p',\psi}
= \|\phi\|^{p'}, 
\een
which gives the desired result. 

Next, we use the Hirschman improvement of the Hadamard three lines theorem \cite{Hirschman,Grafakos}.

\begin{lemma}\label{lem:4}
Let $g(z)$ be holomorphic on the strip $\bS_{1/2}$, continuous and uniformly bounded at the boundary of $\bS_{1/2}$. Then for $\theta \in (0,1/2)$, 
\ben
\ln |g(\theta)| \le \int_{-\infty}^\infty \left(
\beta_{\theta}(t) \ln |g(1+it)|^{2\theta} +  \alpha_{\theta}(t) \ln |g(it)|^{1-2\theta}
\right) \dd t,
\een
where $\alpha_\theta(t),\beta_\theta(t)$ are as in lem. \ref{lem:hirsch}.
\end{lemma}

Applying this to $g=f$ gives the statement of the theorem. 

\medskip

(2) Let us now extend this result to the case where $\rho=\omega_\psi$ is not faithful. 
 We employ the following common trick where we use case (1) above for the modified functional
 \begin{equation}
 \label{rhorep}
 \rho_\epsilon \equiv (1-\epsilon) \rho + \epsilon \sigma ,
 \end{equation}
where $\sigma$ is any faithful normal 
state, which exists since $\mathcal M$ is assumed to be sigma-finite. 
Then $\rho_\epsilon$ in  \eqref{rhorep} with $0 < \epsilon < 1$ is now a faithful state. We take the unique
cyclic and separating vector representative in the natural cone and denote it as $|\psi_\epsilon \rangle$. 

\begin{lemma}
For $1 \le p \le 2$ and $\rho_\epsilon$ the family of states  \eqref{rhorep}, we have $\lim_{\epsilon \to 0^+} \|\zeta\|_{p,\psi_\epsilon}
=\|\zeta\|_{p,\psi}$. 
\end{lemma}
\begin{proof}
Since $\rho_\epsilon /(1-\epsilon) > \rho$, it follows that $\Delta_{\psi_\epsilon, \chi} \ge (1-\epsilon) \Delta_{\psi,\chi}$. 
Therefore, by standard properties of the modular operator, $\Delta_{\chi,\psi_\epsilon}^{-1} \ge (1-\epsilon) \Delta_{\chi,\psi}^{-1}$. 
By L\" owner's theorem \cite{Hansen} applied to the operator monotone (for $1 \le p \le 2$) function $f(x)=x^{1/p-1/2}$, we have
\begin{equation}
\label{domcond}
 \Delta_{\chi,\psi_\epsilon}^{1/2-1/p} 
\geq (1-\epsilon)^{1/p-1/2} \, \Delta_{\chi,\psi}^{1/2-1/p}.
 \end{equation}
Taking the infimum on \eqref{domcond} gives
 \begin{equation}
 \left\| \zeta \right\|_{p,\psi} \leq  \inf_{ \substack{ \chi \in \mathscr{H}: \|\chi\| = 1, \pi(\chi) \geq \pi(\zeta) \\
\zeta \in \mathscr{D}(\Delta_{\chi,\psi_\epsilon}^{1/2-1/p}) } }
\| \Delta_{\chi,\psi}^{1/2-1/p} \zeta \| 
\leq  (1-\epsilon)^{1-2/p} \left\| \zeta \right\|_{p,\psi_\epsilon} .
 \end{equation}
The first inequality holds because the domain restriction gives a smaller class of states over which one takes the infimum
and the second inequality is \eqref{domcond}. 
We therefore obtain 
\ben \label{eq:forward}
\| \zeta \|_{p,\psi}^2 - \| \zeta \|_{p,\psi_\epsilon}^2 \le O(\epsilon). 
\een
Now we use a variational characterization of the $L_p$-norms proven in paper II, prop. 1, for
$1 \le p \le 2$. 
\ben
\| \zeta \|_{p,\psi}^2
=
-\frac{\sin(2\pi/p)}{\pi}
\inf_{x:{\mathbb R}_+ \to {\mathcal M}'} \int_0^\infty [\|x(t)\zeta\|^2  + t^{-1} F_{{\mathcal M}'}(y(t) \omega_\zeta' y(t)^*,\omega_\psi')^2] t^{-2/p'} \dd t , 
\een
where $y(t)=1-x(t)$, the infimum is taken over all step functions $x: {\mathbb R}_+ \to \mathcal{M}'$ with finite range such that $x(t)=1$ for $t \in [0,c]$ for some $c>0$,
and $x(t)=0$ for sufficiently large $t$. We also use the notation $(x\omega x^*)(b)=\omega(x^* a x)$. For any fixed $\delta>0$ a step 
function may be chosen so that the infimum is achieved up to $\delta$. It follows that, with this choice,
\ben \label{eq:reverse}
\begin{split}
& \| \zeta \|_{p,\psi_\epsilon}^2 - \| \zeta \|_{p,\psi}^2 \\
& \le \delta
-\frac{\sin(2\pi/p)}{\pi}
\int_c^\infty [F_{{\mathcal M}'}(y(t) \omega_\zeta' y(t)^*,\rho_\epsilon)^2
- F_{{\mathcal M}'}(y(t) \omega_\zeta' y(t)^*,\rho)^2
] t^{-1-2/p'} \dd t  \\
& \le \delta +  2\, {\rm sinc}(2\pi/p') \, c^{-2/p'} \left( \sup_{t \ge c} \| y(t) \zeta \|^2 \right) \| \rho-\rho_\epsilon \|^{1/2} \le 
\delta + O(\epsilon^{1/2}) , 
\end{split}
\een
using the continuity of the fidelity, lem. \ref{lem:Fcont}  together with $\| \psi-\psi_\epsilon\| \le \| \rho-\rho_\epsilon \|^{1/2}$ from \eqref{eq:PS}
 in the second step, and using \eqref{rhorep} 
in the third step. If we chose $\epsilon$ so small that the $\epsilon$-dependent terms in 
\eqref{eq:forward}, \eqref{eq:reverse} are each less than $\delta$, we get $| \, \| \zeta \|_{p,\psi_\epsilon}^2 - \| \zeta \|_{p,\psi}^2 |<2\delta$. 
Therefore, since $\delta>0$ can be arbitrarily small,  the lemma is proven.
 \end{proof}
Using this lemma  in conjunction with \cite{AM}, lem. 6 (2) gives $\left\| \zeta \right\|_{p,\psi} \le \| \zeta \|$. Then, 
since $|G(z)\rangle$ is assumed to be bounded in the Hilbert space norm, we have $\left\| G(z) \right\|_{p,\psi} \le C$ inside 
the closed strip $\{ 0 \le {\rm Re} z \le 1\}$. Now taking the limit $\epsilon \to 0$ of case (1) for the vector $|\psi_\epsilon\rangle$
using the lemma and the dominated convergence theorem to take the limit under the integral in \eqref{himp}
concludes the proof of (2).  
 \end{proof}


\section{An alternative strategy for proving thm. \ref{thm1}}
\label{alterapproach}

It is conceivable that our approach based on the vector \eqref{ourvec} can be modified by choosing other interpolating vectors, and this may lead 
to new insights relating the argument to somewhat different entropic quantities. Here we sketch an approach which seems 
to avoid the use of $L_p$-norms, thus leading potentially to a substantial simplification. To this end, we consider now a vector 
\begin{equation}\label{xivector}
\ket{{\Xi}_{\psi}(z,\phi)}=\rmo{\psi}{\xi;\mB}^{z}\rmo{\eta}{\xi;\mB}^{-z}\rmo{\eta}{\phi;\mA}^{z}\ket{\psi},
\end{equation}
similar to vectors considered in \cite{CecchiniPetz}.  Here, 
$|\xi\rangle$ is some vector such that $\pi^{\mB^{\prime}}(\xi)\supset \pi^{\mB^{\prime}}(\psi)$, and where in this 
appendix we find it more convenient to think of $\mB$ as defined on the same Hilbert space as $\cA$. The vector \eqref{xivector}
does not depend on the precise choice of $\ket{\xi}$ (but on the vector $\ket{\eta}$ in the natural cone of $\cA$, although we suppress this).

\eqref{xivector} is defined a priori only for imaginary $z$. But if we consider the set of states majorizing $|\psi\rangle$, defined as 
$\sC(\psi, \cA')=\{ |\phi \rangle \in \sH : \| a' \psi\| \le c_\phi \|a' \phi\| \ \ \forall a' \in \cA'\}$, then for $|\phi\rangle$ in this dense linear subspace of $\sH$, it has an analytic continuation to the half strip $\bS_{1/2}=\{ 0<{\rm Re} z < 1/2\}$ that is weakly continuous on the boundary. 
This can be demonstrated by the same type of 
argument as in \cite{CecchiniPetz}, prop. 2.5, making repeated use of the following lemma by \cite{CecchiniPetz}, lem. 2.1:

\begin{lemma}\label{proveanalyticity}
Suppose $|G(z)\rangle$ is a vector valued analytic function for $z \in \bS_{1/2}$, and  $A$ is a self-adjoint positive operator. 
Then $A^z |G(z)\rangle$ is an analytic function of $z \in \bS_{1/2}$ if $\|A^z G(z)\|$ is bounded on the boundary of $\bS_{1/2}$. 
\end{lemma}

For example, we may write 
$\rmo{\eta}{\phi;\mA}^{z}\ket{\psi}=  \rmo{\eta}{\phi;\mA}^{z}\rmo{\psi}{\phi;\mA}^{-z}
\rmo{\psi}{\phi;\mA}^{z}\mo{\psi;\mA}^{-z} \ket{\psi}$, at first for imaginary $z=it$. Using the relations 
\eqref{rels}, \eqref{coswitch}, 
$u'(z)=\rmo{\psi}{\phi;\mA}^{z}\mo{\psi;\mA}^{-z} =(D\psi:D\psi)_{-i\bar z;\cA'}^*$ is a Connes-cocycle for $\cA'$. 
The condition $|\phi\rangle \in \sC(\psi, \cA')$ ensures that it has an analytic continuation from $z=it$ to $\bS_{1/2}$,
as an element of $\cA'$ that is strongly continuous on the boundary of $\bS_{1/2}$ -- this is standard and a proof proceeds as that of lem. \ref{lem:newgamma}, (1). 
Similarly, $v(z)=  \rmo{\eta}{\phi;\mA}^{z}\rmo{\psi}{\phi;\mA}^{-z} = (D\eta:D\psi)_{-iz, \cA}$ is a Connes-cocycle for $\cA$. 

Then, for imaginary $z=it$ we get $\rmo{\eta}{\phi;\mA}^{z}\ket{\psi} = u'(z)v(z) \ket{\psi}$, which has an analytic continuation to 
$\bS_{1/2}$ as $v(z) \ket{\psi}$ is analytic there by Tomita-Takesaki theory. One next applies the lemma with 
$\ket{G(z)}=\rmo{\eta}{\phi;\mA}^{z}\ket{\psi}$ and $A^z = \rmo{\eta}{\psi;\mB}^{-z}$ (chosing $\ket{\xi}=\ket{\psi}$ here). The conditions are verified 
using standard relations of relative Tomita-Takesaki theory as given e.g. in \cite{AM}, app. C, such as \eqref{rels}, \eqref{coswitch}: 
At the upper boundary, $z=1/2+it$, one finds
$\ket{G(1/2+it)}=u'(1/2+it) J_\cA v(it)^* J_\cA \ket{\eta}$ which is of the form $b' \ket{\eta}$ for $b' \in \cA' \subset \mB'$, and one finds 
$A^{1/2+it} = \Delta_{\psi,\eta;\mB'}^{it} J_{\mB'} S_{\psi,\eta;\mB'}$. Together, this gives,
\ben
A^{1/2+it} \ket{G(1/2+it)} = \Delta_{\psi,\eta;\mB'}^{it} J_{\mB'} [u'(1/2+it) J_\cA v(it)^* J_\cA ]^*\ket{\psi}, 
\een 
which is bounded for real $t$. On the other hand, at the lower boundary $A^{it}\ket{G(it)}$ is bounded by definition.
Continuing this type of argument gives the following lemma.

\begin{lemma}\label{xianalyticup}
For  $\ket{\phi} \in \sC(\mA^{\prime},\psi)$, $\ket{{\Xi}_{\psi}(z,\phi)}$ is analytic in the interior of the strip $\bS_{1/2}$ and weakly continuous on the boundary.
\end{lemma}

The relationship with other approaches can be seen through the quantity 
\beq
g(z)=\inf_{\ket{\phi} \in \sC(\mA^\prime,\psi), \|\phi\|=1} \|\Xi_{\psi}(z,\phi)\|.
\eeq
In the setup of finite-dimensional v. Neumann subfactors described in 
sec. \ref{finitealgebra}, we can write
\beq
\ket{\Xi_{\psi}(z,\phi)}=(\rho_{\mB}^z\sigma_{\mB}^{-z}\otimes 1_{\mathcal{C}})\sigma_{\mA}^{z}\rho_{\mA}^{1/2}\tau_{\mA}^{-z}
\eeq
If we take $z=\theta$ real then the infimum over $\tau_\mA$ (the density matrix representing $\ket{\phi}$) readily yields an $L_p$-norm for $p_\theta = 2/(2\theta+1)$,
\beq
g(z)= \( \tr\left|(\rho_{\mB}^z\sigma_{\mB}^{-z}\otimes 1_{\mathcal{C}})\sigma_{\mA}^{z}\rho_{\mA}^{1/2}\right|^{p_\theta}\)^{1/p_{\theta}}.
\eeq
We recognize this again as \eqref{finitergred} corresponding to an expression also studied by \cite{Junge}.

The strategy is now the following. First, lem. \ref{lem:4} also applies to the holomorphic Hilbert-space valued function 
$\ket{\Xi_{\psi}(z,\phi)}$ (because $z \mapsto \ln \|\Xi_{\psi}(z,\phi) \|$ is subharmonic). 
So we have for $0<\theta<1/2$ that
\beq\label{eq:Hirsch}
\ln \|\Xi_\psi(\theta,\phi)\|\leq 
\int_{-\infty}^{\infty}\(\alpha_{\theta}(t)\ln\|\Xi_\psi(it,\phi)\|^{1-2\theta} +\beta_{\theta}(t)
\ln\|\Xi_\psi(1/2+it,\phi)\|^{ 2\theta }\) \ud t.
\eeq
Since $\forall t\in \mathbb{R}, ~\|\Xi_\psi(it,\phi)\|\leq 1,~\alpha_{\theta}(t)>0$, we can drop the first term under the integral. 
Then, we want to divide by $\theta$ and take the infimum over $\ket{\phi} \in \sC(\mA^\prime,\psi), \|\phi\|=1$. The next lemma will allow us to deal with the second term under the integral. Since $\ket{\phi} \in \sC(\cA^\prime,\psi)$, we can write $\ket{\psi}=a\ket{\phi}$, where $a \in \cA$ is self-adjoint, see \cite{Zsido}, 5.21. Then:

\begin{lemma}\label{prooffid} We have
\beq\label{proofrmo}
\| \Xi_{\psi}(1/2+it,\phi) \|^2 = \omega_{\psi}\circ\iota\circ \alpha_\eta^t(a^2)
\eeq
for all $\ket{\phi} \in \sC(\mA^\prime,\psi)$.
\end{lemma}

\begin{proof}
On the left hand side of \eqref{proofrmo}, we may choose $\ket{\xi}=\ket{\eta}$. It is most convenient to work with 
state vectors in the natural cones, for notations see \eqref{eq:notation}. Define 
$b=\rmo{\psi}{\eta;\mB}^{1/2}\mo{\eta;\mB}^{-1/2}$, which is affiliated to the algebra $\mB$ and extend the definition \eqref{petzdef} to affiliated operators. Then we can write
\beq
\begin{split}
& \| \Xi_{\psi}(1/2+it,\phi) \|^2\\
=&\|\rmo{\psi}{\xi;\mB}^{1/2}\rmo{\eta}{\xi;\mB}^{-\frac 12-it}\rmo{\eta}{\phi;\cA}^{\frac 12+it} a \phi_\cA \|^2 \\
=&\| J_{\mB} \varsigma_{\eta;\mA}^{-t}\(\varsigma_{\eta;\mB}^{t}\(b\)\)J_{\mA} a\eta_\cA \|^2 \\
=&\bra{\eta_\cA }\varsigma_{\eta;\mB}^{t}(b^* b)J_{\mA}\varsigma_{\eta;\mA}^{t}\( a^2\)\ket{\eta_\cA}\\
=&\bra{\eta_\mB}\varsigma_{\eta;\mB}^{t}(b^*b)J_{\mB}\alpha_{\eta}(\varsigma_{\eta;\mA}^{t}\( a^2\))\ket{\eta_\mB}\\
=&\bra{\eta_\mB}b^*b J_{\mB}\varsigma_{\eta;\mB}^{-t}\(\alpha_{\eta}\(\varsigma_{\eta;\mA}^{t}\( a^2\)\)\)\ket{\eta_\mB}\\
=&\bra{\eta_\mB}b^* J_{\mB}\varsigma_{\eta;\mB}^{-t}\(\alpha_{\eta}\(\varsigma_{\eta;\mA}^{t}\( a^2\)\)\) b \ket{\eta_\mB}\\
=&\bra{\psi_\mB}\alpha_{\eta}^{t}(a^2)\ket{\psi_\mB}=\omega_{\psi}\circ\iota\circ \alpha_\eta^t(a^2).
\end{split}
\eeq
(The choice $\pi^{\mA^\prime}(\psi)=J_{\mA}^2\leq\pi^{\mA^\prime}(\phi)\leq\pi^{\mA^\prime}(\eta)=1$ guarantees the supports of vectors on $\cA^\prime$ are multiplied in the correct way, so we keep the $\pi^{\cA^\prime}$'s implicit in the derivation -- everything should be understood to happen on $\pi^{A^\prime}(\psi)$.) In the derivation we used the definition of the Petz recovery map, see e.g. \cite{Petz1993} proof of prop. 8.4,
such that  $\forall~a \in \mA, b \in \mB$,
\beq\label{petzdef}
\bra{\eta_\cA} b J_{\mA}  a \ket{\eta_\cA}=\bra{\eta_\mB} b J_{\mB} \alpha_{\eta}(a)\ket{\eta_\mB}.
\eeq
Thus, we have \eqref{proofrmo}. We obtain the claim in the lemma by taking the infimum in the set $\sC(\mA^\prime,\psi)$ on both sides of \eqref{proofrmo} and using \eqref{uhl1}. 
\end{proof}

The lemma and concavity of $\ln$ allows us to conclude from \eqref{eq:Hirsch} that 
\beq\label{eq:Hirsch1}
\lim_{\theta \to 0^+} \frac{1}{\theta} \ln \|\Xi_\psi(\theta,\phi)\|\leq 
\ln \| a \zeta_\sS \|^2 = \ln \| \Delta^{1/2}_{\zeta_\sS,\phi} \psi\|^2,
\eeq
where $\ket{\zeta_\sS}$ is a vector representative of $\omega_\psi \circ \iota \circ \alpha_\sS \in \cA_\star$ and 
$\alpha_\sS$ the recovery channel \eqref{recover}. Note that taking the infimum over $\ket{\phi} \in \sC(\psi, \cA')$ on the right 
side yields $2 \ln F(\omega_\psi, \omega_\psi \circ \iota \circ \alpha_\sS)$
On the other hand, it is plausible  to expect that for the term on the left side of \eqref{eq:Hirsch}, we obtain
\beq
\label{eq:conj}
\inf_{\phi\in \sC(\mA^\prime, \psi)}  \lim_{\theta\to 0^+}\frac{1}{\theta}\ln\|\Xi_{\psi}(\theta,\phi)\|=-S_\cA(\psi|\eta)+S_\mB(\psi | \eta). 
\eeq
If this latter equation could be demonstrated -- which is possible at a formal level\footnote{
It is relatively straightforward to see that this equation would follow from the equation
\beq\label{31}
\lim_{\theta \to 0^+}\frac{1}{2\theta}\(1-\|\Xi_{\psi}(\theta,\phi)\|^2\)=\bra{\psi}\ln\rmo{\eta}{\psi;\mB}\ket{\psi}-\bra{\psi}\ln\rmo{\eta}{\phi;\mA}\ket{\psi}.
\eeq
which is easier to check as it does not contain an infimum.
} --
then it is clear that we would obtain an alternative proof of thm. \ref{thm2} (though not of thm. \ref{thm1}). 

When attempting to demonstrate \eqref{eq:conj} (or equivalently \eqref{31}), 
one is facing similar technical difficulties as in the proof strategy described in the 
body of the text. There, we were forced to introduced suitably regularized versions $|\psi_P\rangle$ of the vector in question. Thus, 
while the strategy discussed in this appendix nicely avoids the use of $L_p$-spaces up to a certain point, 
it is not clear whether their use can be altogether avoided. We think that this would be an interesting research project.


\end{document}